\documentclass[a4paper]{article}
\pdfoutput=1

\usepackage{natbib}
\usepackage{amsmath,amssymb,amsthm,bbm}
\usepackage{mathtools}

\theoremstyle{plain}
\newtheorem{theorem}{Theorem}
\newtheorem{corollary}[theorem]{Corollary}
\newtheorem{lemma}[theorem]{Lemma}
\newtheorem{conjecture}{Conjecture}
\newtheorem{openproblem}{Open problem}

\theoremstyle{definition}
\newtheorem*{definition}{Definition}
\newtheorem{axiom}{Axiom}

\DeclareRobustCommand{\sortkey}[1]{}

\DeclareMathOperator{\bigO}{\mathcal{O}}
\DeclareMathOperator{\dom}{dom}
\DeclareMathOperator{\gap}{gap}
\DeclareMathOperator{\ic}{ic}

\DeclareMathOperator{\pad}{pad}

\newcommand{\bbN}{\mathbbm{N}}
\newcommand{\binary}{\mathbbm{2}}
\newcommand{\calF}{\mathcal{F}}
\newcommand{\calL}{\mathcal{L}}
\newcommand{\cl}[1]{\textnormal{\textbf{#1}}}	% Complexity classes
\newcommand{\clnu}[1]{\cl{#1$_\textbf{nu}$}}	% Nonuniform complexity classes
\newcommand{\clsu}[1]{\cl{#1$_\textbf{u}$}}	% Semi-uniform complexity classes (subscript `su' is used by some authors to denote strong uniformity)
\newcommand{\clO}[1]{\ensuremath{\boldsymbol{\bigO(}#1\boldsymbol{)}}}	% Big-O complexity classes
\newcommand{\clX}[1]{\cl{\textit{X}#1}}	% Slicewise complexity classes
\newcommand{\clXsu}[1]{\clsu{\textit{X}#1}}	% Semi-uniform slicewise complexity classes
\newcommand{\clXnu}[1]{\clnu{\textit{X}#1}}	% Nonuniform slicewise complexity classes
\newcommand{\deq}{=}	% Defined equal
\newcommand{\length}[1]{{\lvert#1\rvert}}
\newcommand{\lli}{linearly length increasing}
\newcommand{\nquasile}{\not\preccurlyeq}
\newcommand{\quasile}{\preccurlyeq}
\newcommand{\reland}{\;\land\;}	% Relating conditions conjunctively
\newcommand{\relor}{\;\lor\;}	% Relating conditions disjunctively
\newcommand{\rmnu}{\mathrm{nu}}

\newcommand{\st}[1][]{{\;#1{|}\;}}	% Such that (with optional size)
\newcommand{\symdiff}{\bigtriangleup}

\newcommand{\immune}[1]{almost #1-bi-immune}
\newcommand{\levelable}[1]{#1-omni-levelable}

\title{Levelable Sets and the Algebraic Structure of Parameterizations}
\author{Jouke Witteveen \and Leen Torenvliet}

\begin{document}

\bibliographystyle{plainnat}

\maketitle

\begin{abstract}
  Asking which sets are fixed-parameter tractable for a given parameterization constitutes much of the current research in parameterized complexity theory.
  This approach faces some of the core difficulties in complexity theory.
  By focussing instead on the parameterizations that make a given set fixed-parameter tractable, we circumvent these difficulties.
  We isolate parameterizations as independent measures of complexity and study their underlying algebraic structure.
  Thus we are able to compare parameterizations, which establishes a hierarchy of complexity that is much stronger than that present in typical parameterized algorithms races.
  Among other results, we find that no practically fixed-parameter tractable sets have optimal parameterizations.
\end{abstract}

\section{Introduction}

Ever since the identification of efficient computability by Cobham and Edmonds (see~\citealp{goldreich2008computational,arora2009computational}) and subsequent intractability results \citep{cook1971complexity,garey1979computers}, the computational complexity of sets has been focused on the complexity of their hardest instances, or in the case of average computational complexity on the complexity of the majority of their instances.
Yet, even very hard sets have simple instances and often lots of them.

Parameterized complexity \citep{downey1999parameterized, flum2006parameterized} was originally introduced to deal with this apparent indiscriminate judgment of computational complexity of sets.
Instead of looking at the entire set as a single computational object, the complexity of the set is stratified.
A parameter function is introduced that singles out one particular dimension of the set, such as  the size of the desired solution, and a function of this parameter is factored into the complexity of the computation.
Then, for any part of the set where this function yields a constant, membership of the set may be efficiently computable.
Sets that can thus be dissected into parts that are decidable in polynomial time, where the degree of the polynomial is invariable across all parts, form the parameterized counterpart of efficient computability, \emph{fixed-parameter tractability}.
This division into polynomial time computable parts may be uniform or nonuniform (i.e., with or without a single piecewise polynomial time algorithm) and the function of the parameter in play may be computable or non-computable.
Thus, sets can be divided up into parts that are easy or tractable from one of many viewpoints.

At the same time, many sets have infinite subsets that do not have tractable algorithms.
\citet{lynch1975reducibility} introduced \emph{complexity cores} and showed that intractable sets have infinite subsets of which all but finitely many instances resist all efficient programs.
In particular no parameter function of a fixed-parameter tractable set can take on the same value infinitely often on a complexity core.

Unfortunately, complexity cores allow finite variations and thus any specific instance can arbitrarily be made part of, or excluded from a complexity core.
Therefore, complexity cores do not lead to a useful formalization of the idea of hard instances.
Exclusion of an instance from a complexity core is possible by augmenting a program with table lookup for that instance.
This shortcoming can be overcome by taking also the sizes of programs into account.
Resource bounded \emph{instance complexity}, as introduced by \citet{orponen1994instance}, does exactly that as a combination of the complexity of the set and the individual complexity of strings \citep{li1997introduction}.

Both instance complexity and parameterized complexity can be seen as ways of extending complexity cores into a meaningful notion of the distribution of complexity inside a set.
While instance complexity focuses on sets directly, parameterized complexity instead focuses on algorithms, or, more abstractly, on the possible ways to slice up a set into tractable parts.
In general, there may be many different algorithms for deciding a particular set and in that regard instance complexity is less ambiguous than parameterized complexity.
However, any finite number of algorithms can be combined into a single one that boasts the best behavior of its constituents.
Such a compound algorithm would be structured as a big if-then-else case distinction.
On the parameterization side, this means that any finite number of parameters can be combined into a single parameter via an easy to compute pairing function.
This compound parameter then represents the complexity obtained when taking all of the finitely many algorithms considered into account.

Besides treating parameterizations as measures of complexity, it may be tempting to view a classification of a set as fixed-parameter tractable as information about the computational complexity of that set.
However, \citet[Corollary~3.8]{witteveen2016fpd} have shown that the class of fixed-parameter tractable sets equals that of fixed-parameter decidable sets.
Hence classifying a set as fixed-parameter tractable holds very little information about its computational complexity.
When imposing additional uniformity constraints, the class of fixed-parameter tractable sets quickly becomes equal to that of the decidable sets.
Again, a classification as fixed-parameter tractable is of little value from a computational complexity point of view.
We therefore consider it more informative to study parameters and parameterizations first, and sets second.

In studying parameterizations, we should not ask what sets are fixed-parameter tractable with a given parameterization, even though such questions constitute much of the current research in parameterized complexity theory (e.g.,~\citealp{bodlaender2013deterministic}).
The question ``is some set $A$ (uniformly) fixed-parameter tractable with a given parameterization?'' is equivalent to ``are the parts of $A$ with a fixed parameter value (uniformly) decidable in polynomial time?''
If this was in general an easy question, we would not have so much trouble separating \cl{P} from other complexity classes.
Instead, we should explore the parameterizations with which a given set is fixed-parameter tractable (e.g.,~\citealp{garg2016raising}).
This circumvents the difficulties associated with having to separate complexity classes.

We found that in order to study parameterizations properly, we need a definition of parameterizations that is more general than most common definitions.
Often, parameters are represented by integer values associated with instances.
Whether a parameter represents the size of a vertex cover in a graph, a bound on the treewidth in a graph, or the number of classrooms in a scheduling problem, it is usually a function that yields a number.
It has been recognized that parameters need not be restricted to one dimensional numerical values, yet to our knowledge the true limits of parameters have not been investigated before.
In Section~\ref{sec:parameterspaces}, we shall identify an axiomatic basis of traits required from parameters and parameterizations.
This allows us to build general definitions from the ground up, enforcing just the properties we desire.
The resulting definitions are rooted in order~theory and reveal a rich algebraic structure that governs parameterizations.
Notions from order theory required for our definitions are defined also in Section~\ref{sec:parameterspaces}. % Directed set, up-set, lattice, filter
Because of the role of order theory in this paper, we meet various kinds of orders that (almost) all have their own symbol.
There is $\leq$, which we reserve for natural numbers, $\bbN$.
In text, we refer to this as ``smaller''.
We use $\quasile$ both for orders on parameter spaces and for orders on parameterizations, and refer to these orders as ``before'' and ``below'' respectively.
Finally, sets of parameterizations take center stage in this paper.
For such sets, the usual set inclusion ordering $\subseteq$ is in place.
When the parameterizations with which a given set is fixed-parameter tractable form a subset of the parameterizations with which some other set is fixed-parameter tractable, we could say that the latter set is ``easier'' than the former.

Historically, parameterized complexity theory is, at least in part, practically motivated \citep{downey1999parameterized}, and its definitions are arrived at in a pragmatic fashion.
In Sections~\ref{sec:nonuniform} and~\ref{sec:uniform}, continuing our foundational journey, we recover the well-known parameterized complexity classes \clX{P} and \cl{FPT} from classical complexity theory primitives and our general definition of parameterizations.
The first of these sections deals with nonuniform parameterized complexity, whereas the second deals with semi-uniform and fully uniform parameterized complexity.
Section~\ref{sec:preliminaries} contains the required background in classical complexity theory, as far as used for building parameterized complexity theory bottom-up.
Chiefly, that section includes the definition of a \emph{slice}, which serves as a dual to the complexity \emph{core}.

Much of our parameterized analysis of complexity revolves around collections of parameterizations that put a given set in a parameterized complexity class.
Such collections function as an interface to the complexity of the sets that gave rise to them.
In Sections~\ref{sec:nonuniform} and~\ref{sec:uniform} we lift classical levelability and immunity classifications to the parameterized context.
The parameterized versions of these classifications can be expressed as properties of the collection of parameterizations that put a given set in a parameterized complexity class.
Doing so makes clear that levelability captures a usefulness criterion for parameterized algorithms.
Intuitively, a parameterized approach to decision algorithms is only useful when there are infinitely many instances corresponding to any fixed parameter value.
Levelability with respect to parameterized complexity classes is the formalization of this intuition.
Another complexity aspect captured by the collection of parameterizations that put a given set in a parameterized complexity class is the existence of optimal parameterizations.
As it turns out, these collections of parameterizations form filters in a lattice of parameterizations.
A parameterization that is better than (or in terms of this paper \emph{below}) all other parameterizations exists precisely when such a filter is a principal filter.
We obtain an almost complete characterization of the sets that have optimal parameterizations with respect to \cl{FPT} and \clX{P} with varying uniformity constraints.
From this characterization we get that uniformly fixed-parameter tractable sets for which a parameterized approach is useful admit no optimal parameterizations.
A completion of the characterization and two conjectures regarding the separations of sets based on their complexity are left as open problems.

\section{Preliminaries}
\label{sec:preliminaries}

We assume the reader is familiar with standard notation from complexity theory, as in \citep{papadimitriou2003computational,arora2009computational}.
In the current text we use a binary alphabet $\binary \deq \{0, 1\}$.
The set of nonempty finite binary sequences, \emph{strings}, is denoted by $\binary^+$.
Complexity classes are written in boldface.
We make use of $\bigO$-notation, where $n$ is the free variable in the mathematical expression following $\bigO$.
For a function $f$, the complexity class \clO{f(n)} is the class of sets that can be decided deterministically by a Turing machine with a running time bound in $\bigO(f(n))$ using some agreed upon alphabet and number of tapes.

Of interest to us are decision procedures that not necessarily decide on membership of every element.
\begin{definition}
  A Turing machine $\Phi$ is a \emph{partial decision procedure} for a set $A$ if we have, for all $x \in \binary^+$:
  \begin{itemize}
    \item $\Phi(x) = 1 \implies x \in A$.
    \item $\Phi(x) = 0 \implies x \notin A$
  \end{itemize}
  The machine is said to \emph{decide} the elements of its \emph{domain}:
  \begin{equation*}
    \dom(\Phi) \deq \{x \st \Phi(x) \in \{1, 0\}\}.
  \end{equation*}
\end{definition}

Outside its domain, a partial decision procedure either does not halt at all, or outputs anything other than $0$ or $1$.
Limiting partial decision procedures to a polynomial running time gets us polynomial approximations as in {\citep{ko1981completeness,balcazar1985bi}}.
\begin{definition}
  A procedure $\Phi$ is a \emph{\cl{P}-approximation} for a set $A$ if it is a partial decision procedure for $A$ that runs in polynomial time.
\end{definition}

Note that although a \cl{P}-approximation halts on every input, we do not demand that its domain is $\binary^+$.
Necessarily, however, the domain of a \cl{P}-approximation is in \cl{P}.

Much of our work revolves around sets that occur as domains of partial decision procedures.
Therefore, we shall introduce a name for such sets.
\begin{definition}
  A set $S$ is a \emph{\cl{P}-slice} for a set $A$ if there is a \cl{P}-approximation $\Phi$ for $A$ satisfying $\dom(\Phi) = S$.
\end{definition}

The name \emph{slice} is inspired by the use of that term in parameterized complexity theory.
Slices function as a dual to the complexity cores of \citet{lynch1975reducibility}.
\begin{definition}
  A set $C$ is a \emph{\cl{P}-core} for a set $A$ if for every \cl{P}-slice $S$ for $A$ the intersection $C \cap S$ is finite.
\end{definition}

Subsets of \cl{P}-cores for a set are also \cl{P}-cores for that set, as are finite variations.
This complicates thinking of the members of a core for a set as the inherently hard instances of that set.
However, for some sets the collection of \cl{P}-cores contains a maximal element with respect to inclusion up to finite variation.
Such sets are split into an easy part and a hard part.
\begin{theorem}
\label{thm:maximal}
  A set $A$ has a maximal (up to finite variations) \cl{P}-slice if and only if it has a maximal \cl{P}-core.
\end{theorem}
\begin{proof}
%\begin{description}
%  \item[$\implies$]
  $\Longrightarrow$.
    The complement of a maximal \cl{P}-slice for $A$ is a \cl{P}-core for $A$ and cannot be infinitely extended.
    Hence it is a maximal \cl{P}-core for $A$.

%  \item[$\impliedby$]
  $\Longleftarrow$.
    We claim that the complement of a maximal \cl{P}-core for $A$ is a \cl{P}-slice for $A$.
    Let $S_1, S_2, \ldots$ be an enumeration of the \cl{P}-slices for $A$.
    Assuming that the complement of a maximal \cl{P}-core for $A$ is not a \cl{P}-slice for $A$, then for all $j$, there are infinitely many elements in this complement outside the \cl{P}-slice $\bigcup_{i \le j} S_i$.
    Consequently, we would be able to extend the maximal \cl{P}-core with infinitely many elements, one for each $j$, contradicting its maximality.
    Hence, the complement of a maximal \cl{P}-core for $A$ must be a \cl{P}-slice for $A$.
%\end{description}
\end{proof}

Thus the complement of a maximal \cl{P}-slice is a maximal \cl{P}-core and because \cl{P}-slices were necessarily in \cl{P} we get the following.
\begin{corollary}
  A maximal \cl{P}-core, if it exists, is in \cl{P}.
\end{corollary}
Of course, within a maximal \cl{P}-core $C$ for a set $A$, membership of $A$ is hard to decide in the sense that no \cl{P}-approximation decides membership of $A$ for infinitely many elements from $C$.

Complexity cores give us an elegant definition of immunity that is nonstandard, but equivalent to the common definition in the literature based on finite subsets.
The equivalence of both definitions was already observed by \citet{balcazar1985bi} (see also~\citealp{book1988polynomial}).
\begin{definition}
  A set $A$ is \emph{\cl{P}-bi-immune} if $\binary^+$ is a \cl{P}-core for $A$.
\end{definition}

Thus a set is \cl{P}-bi-immune if it has the largest \cl{P}-core possible.
This definition is easily generalized to sets that have a maximal \cl{P}-core.
In this case the equivalence of our definition to the common definition as a disjoint union was observed by \citet{orponen1986classification}.
\begin{definition}
  A set is \emph{\immune{\cl{P}}} if it has a maximal \cl{P}-core.
\end{definition}
Note that by our definitions sets in \cl{P}, and finite sets in particular, have maximal \cl{P}-cores and are therefore \immune{\cl{P}}.
This may seem somewhat peculiar, but is in agreement with the definitions used by \citet{orponen1986classification} and of no objection in the sequel.

With \cl{P}-slices and \cl{P}-cores we make no distinction between the members and the nonmembers of a set.
For this reason and by Theorem~\ref{thm:maximal} a set is \immune{\cl{P}} if it has a slice that cannot be extended by infinitely many members or nonmembers into a larger slice.
In \citep{orponen1985polynomial,orponen1986optimal} sets of which every slice can be extended by infinitely many members into a larger slice are called \emph{\cl{P}-levelable}.
We want a more general definition covering sets of which every slice can be extended by infinitely many members or nonmembers.
Although it may be tempting to call such sets \cl{P}-bi-levelable, this would more naturally describe sets of which every slice can be extended by both infinitely many members \emph{and} infinitely many nonmembers.
The equivalence of the following definition and that based on infinite extensions of slices was observed by \citet{orponen1985polynomial}.
\begin{definition}
  A set is \emph{\levelable{\cl{P}}} if it has no maximal \cl{P}-slice.
\end{definition}
By Theorem~\ref{thm:maximal}, a set is \levelable{\cl{P}} precisely when it is not \immune{\cl{P}}.
Furthermore, every \cl{P}-levelable set is \levelable{\cl{P}}.
As a consequence of the observation of \citet{orponen1986optimal} that very many (natural) intractable sets are \cl{P}-levelable, we find that there are many \levelable{\cl{P}} sets.

In all of the above definitions and theorems, \cl{P} can be replaced, given some constant $c$, by \clO{n^c}.
Likewise, an \clO{n^c}-approximation for a set $A$ is a partial decision procedure for $A$ with a running time in $\bigO(n^c)$.

\section{Parameter Spaces and Parameterizations}
\label{sec:parameterspaces}

It is customary in computability theory to consider consider convergence of computation \citep{odifreddi1992classical}.
Computation that terminates is said to be \emph{convergent}, while computation that does not halt is said to be \emph{divergent}.
In computational complexity theory the central theme is convergence of computation within a bounded amount of resource (time, space, \ldots) usage.
Parameterized complexity theory originates from the desire to have a more fine-grained analysis of computational complexity, that is, a more fine-grained analysis of convergence of computation.
In the current section we shall adopt this mindset and define parameterizations from the ground up as mathematical structures.

When we are not interested in resource usage, convergence of a Turing machine on a given input is a simple matter of \emph{yes} or \emph{no}.
We can represent convergence of a Turing machine $\Phi$ by the set $\{x \st \text{$\Phi$ converges on $x$}\}$.
In a context where resource usage matters, we are mostly interested in Turing machines that converge on all inputs, but with varying resource needs.
It is these resource needs that we want to measure using parameters.
With the parameter, we want to express how much of some resource is needed for a given Turing machine to converge on a given input.
Thus, we want to represent convergence of a Turing machine $\Phi$ by something like the set
\begin{equation*}
  \{(x, k) \st \text{$\Phi$ converges on $x$ with at most $k$ (additional) resource usage}\}.
\end{equation*}
From the presence of \emph{at most} in the above criteria, it is clear that parameter values have to be taken from ordered sets.
These ordered sets we shall call \emph{parameter spaces} and the sets of tuples $(x, k)$ as above we shall call \emph{parameterizations}.
While we have used a parameter, $k$, to measure resource usage, we have neglected to specify how resources should be measured.
Indeed, this is not something that can be done unambiguously and therefore whenever we say a parameter measures a resource, we mean that the resource is bounded in a parameter dependent way.
The precise dependence (such as, for example, exponential or double exponential) we shall not pay attention to in this paper.

In order to come up with proper definitions for parameter spaces and parameterizations, we consider some of the properties we want them to have.
As ours is the study of computational complexity and not that of computability, we require that it is always possible to have sufficient parametric leeway for a computation to converge.
\begin{axiom}
  For every input $x$, computation converges somewhere (i.e., there is a parameter value with which $x$ occurs in the parameterization).
\end{axiom}

A second trait we put in place is that convergence should be stable.
\begin{axiom}
  When some $x$ is in a parameterization with a parameter value $k$, then, for all $k'$ following $k$ (i.e., where $k$ comes before $k'$ in a parameter space specific order), $x$ is also in the parameterization with parameter value $k'$.
\end{axiom}

When computations produce results, such as is the case for decision procedures, we want these results to not change for some input $x$ depending on the parameter values with which $x$ is in a parameterization.
This is relevant since we shall eventually provide Turing machines access to the value of the parameter.
Demanding that the outcome of a computation is the same for all parameter values following one with which some input occurs in a parameterization is not enough to enforce convergence to a unique outcome.
It could be that the order on the parameter space defines incomparable (with regard to the order) branches of parameter values.
Not only do we want convergence to occur on all these branches for all inputs, we also want some enforcement of the uniqueness of the outcome of a converged computation.
\begin{axiom}
  For every two parameter values $k_1, k_2$ there is a parameter value $k$ such that $k_1$ and $k_2$ come before $k$ in the applicable order.
\end{axiom}
This axiom internalizes both the idea that resources can always be sufficiently extended and the desire for convergence to only be possible to a single outcome.

Lastly, we want to be able to use parameters constructively.
\begin{axiom}
\label{ax:last}
  Parameter values can be effectively encoded so that they can be used in data structures.
\end{axiom}

These \ref{ax:last} axioms can be realized using structures known from order theory \citep{davey2002introduction}.

\begin{definition}
  A nonempty set $\Omega$, ordered by a reflexive and transitive order $\quasile$, is a \emph{directed set} if for all $a, b \in \Omega$ there is a $c \in \Omega$ such that we have $a \quasile c$ and $b \quasile c$.
\end{definition}

Reflexive and transitive orders that are not necessarily antisymmetric are often referred to as \emph{quasiorders}.
A quasiorder that is also antisymmetric is a \emph{partial order}.
In this paper, quasiorders are the dominant type of orders so we shall not drag the term \emph{quasi} around and instead speak simply of orders.

\begin{definition}
  A subset $U$ of a directed set $\Omega$ is an \emph{up-set} if, for all $a \in U$ and $b \in \Omega$ with $a \quasile b$, we have $b \in U$.
\end{definition}

Using these concepts, we can formulate the most general definitions that meet our axioms.

\begin{definition}
  A \emph{parameter space} is a directed set that can be encoded into $\binary^+$.
  The length of an element $k$ of a parameter space is the length of its encoding and denoted by $\length{k}$.
\end{definition}

We call a parameter space \emph{decidable} if the range of its encoding is.
Likely the most prevalent parameter space is the set of natural numbers ordered by the regular \emph{smaller than} order and encoded with ordinary binary representation.

\begin{definition}
  Given a parameter space $\Omega$, a \emph{parameterization} is a subset $\eta$ of $\binary^+ \times \Omega$, such that for every $x$ the set $\{k \st (x, k) \in \eta\}$ is a nonempty up-set of $\Omega$.
  For a fixed parameter value $k$, we denote by $\eta_k$ the set $\{x \st (x, k) \in \eta\}$.  % and $\eta(x) \deq \{k \st (x, k) \in \eta\}$
\end{definition}

Observe how, for a parameterization $\eta$ and a fixed parameter value $k$, the set $\eta_k$ is similar to our initial representation of convergence of a Turing machine $\Phi$ by the set $\{x \st \text{$\Phi$ converges on $x$}\}$.
We are not so much interested in finite variations on the sets of inputs on which computation converges.
Instead, we are mostly interested in parameterizations for which these sets increase in infinitely large steps.

\begin{definition}
  A parameterization $\eta$ has \emph{imix} (infinitely many infinite extensions) if for every parameter value $k$ there is a parameter value $k'$ such that the set $\eta_{k'} \backslash \eta_k$ is infinite.
\end{definition}

This paper is not the first to define parameterizations, and ideally our definition would be equivalent to earlier definitions.
Unfortunately, our axiomatic definitions differ from the definitions of both \citet{downey1999parameterized} and \citet{flum2006parameterized}, which were obtained in a more empirical fashion.

The concept of a parameterization as used by \citet{flum2006parameterized} is closest to ours as it can be considered a special case of our definition.
For \citet{flum2006parameterized}, a parameterization is a polynomial time computable function $\kappa: \binary^+ \to \bbN$ mapping inputs to a unique point of convergence in the parameter space $\bbN$.
This can be viewed as an instance of \emph{parameterizing by a function} $f: \binary^+ \to \Omega$, where $\Omega$ is a parameter space, in general.
The parameterization corresponding to such a function is the set $\{(x, k) \st f(x) \quasile k\}$.
Note that, contrary to \citet{flum2006parameterized}, we have not required that there is a \emph{unique} point of convergence in parameterizations and in particular we have not required that every parameterization arises from parameterizing by some function.
The benefit of our more general definition is that it allows for natural operations combining parameterizations into new ones.

With \citet{downey1999parameterized}, there is no distinction between the sets of which the computational complexity is studied and parameterizations.
Their study is restricted to multi-dimensional sets of which the computational complexity is expressed as a function of the values along each of the dimensions of the sets.
As a result, the notion of convergence of computation is lost since it is no longer possible to speak of the kind of limit behavior we demanded by our axioms.
Where a point of convergence was unique with \citet{flum2006parameterized}, it is nonexistent with \citet{downey1999parameterized}.
In practice, however, most of the sets considered by \citet{downey1999parameterized} are monotone in the sense that when a parameter value $k$ comes before another parameter value $k'$, the inputs associated with $k$ form a subset of those associated with $k'$.
While this restores a point of convergence for members of a set in the form of a first parameter value with which it occurs in the parameterized set, it does no good for nonmembers.
We feel that this asymmetry is undesirable and that parameterized complexity theory should be invariant under taking the complement of a set.

Both \citet{downey1999parameterized} and \citet{flum2006parameterized} work mostly with $\bbN$ as their parameter space, ordered by its canonical order.
Having a common parameter space makes it easy to compare parameterizations \citep{komusiewicz2012new}.
While our more general definitions open the door for parameter spaces with a richer structure than that of $\bbN$, they also necessitate a special means of comparing parameterizations.
It is possible to numerically capture the point of convergence for some input given a parameterization.

\begin{definition}
  Given a parameterization $\eta \subseteq \binary^+ \times \Omega$, the minimization function $\binary^+ \to \bbN$ of $\eta$ is defined as
  \begin{equation*}
    \mu_\eta(x) \deq \min\{\length{k} \st k \in \Omega \reland (x, k) \in \eta\}.
  \end{equation*}
\end{definition}

This minimization function allows us to compare the behavior of the point of convergence of multiple parameterizations with different underlying parameter spaces.
Informally, we want to express how long a parameter value needs to be for convergence in one parameterization to happen on all inputs that converge with some bounded point of convergence in another parameterization.

\begin{definition}
  Given parameterizations $\eta_1$ and $\eta_2$, the gap function $\bbN \to \bbN \cup \{\infty\}$, is defined as
  \begin{equation*}
    \gap_{\eta_1, \eta_2}(n) \deq \max\{\mu_{\eta_1}(x) \st x \in \binary^+ \reland \mu_{\eta_2}(x) \le n\},
  \end{equation*}
  where we take the maximum of the empty set to be $0$.
\end{definition}

Comparing parameterizations using this gap function enables us to define a nonuniform and a uniform order on parameterizations.
One parameterization is below another, when a bound on the point of convergence for the other parameterization can be turned into a bound on the point of convergence for the one parameterization.
Similar orders have been considered by \citet{komusiewicz2012new} and \citet{fellows2013towards}.

\begin{definition}
  A parameterization $\eta_1$ is below a parameterization $\eta_2$ in the \emph{nonuniform} order $\quasile_\rmnu$ if for all $n \in \bbN$ we have $\gap_{\eta_1, \eta_2}(n) < \infty$.
\end{definition}

From a computational standpoint, a bound on the gap between two parameterizations is only useful if it is computable.
The uniform variant of the order on parameterizations is therefore mainly of interest when the parameterizations involved are based on decidable parameter spaces and are themselves decidable.

\begin{definition}
  A parameterization $\eta_1$ is below a parameterization $\eta_2$ in the \emph{uniform} order $\quasile$ if there is a computable function upper bounding $\gap_{\eta_1, \eta_2}$.
\end{definition}

We observe the relationship $\quasile \subset \quasile_\rmnu$ between these orders.
These definitions provide structure to the class of parameterizations.
\begin{lemma}
\label{lem:preorder}
  Both $\quasile_\rmnu$ and $\quasile$ are reflexive and transitive orders on the class of parameterizations.
\end{lemma}
\begin{proof}
  Reflexivity follows from the observation that for every parameterization $\eta$ the gap function $\gap_{\eta, \eta}$ is bounded by the identity function.
  For transitivity we need only to remark that the composition of finite bounding functions is again a finite bounding function and that the composition of computable functions is computable.
\end{proof}

Neither order is antisymmetric and it is convenient to work with the associated partially ordered set of \emph{equivalence classes} instead of with parameterizations directly.
Note that every parameterization with imix is unequal, in the associated equivalence relation, to any parameterization without imix, both using the nonuniform as well as using the uniform order.

\begin{lemma}
\label{lem:imix}
  Given parameterizations $\eta$ and $\eta'$, when $\eta$ has imix but $\eta'$ does not, then one of $\eta \quasile_\rmnu \eta'$ and $\eta' \quasile_\rmnu \eta$ fails to hold (and similarly for $\quasile$).
\end{lemma}
\begin{proof}
  The statement for the uniform order follows from that for the nonuniform order by the inclusion $\quasile \subset \quasile_\rmnu$.

  Let $m$ be so that no parameter value for $\eta'$ of length at least $m$ has an infinite extension.
  Because parameter spaces are directed sets, for every parameterization that does not have imix such an $m$ can be found.

  Suppose we have $\eta \quasile_\rmnu \eta'$.
  There is then a parameter value $k$ such that $\eta_k$ is a superset of all sets $\eta'_{k'}$, when we have $\length{k'} \le m$.
  Since $\eta$ has imix, there exists a parameter value $l$ such that $\eta_l \backslash \eta_k$ is infinite.
  We find that $\eta' \quasile_\rmnu \eta$ fails because $\gap_{\eta', \eta}(\length{l})$ must be infinite.

  In case we have $\eta' \quasile_\rmnu \eta$, then if $\eta \quasile_\rmnu \eta'$ were to hold as well, $\gap_{\eta, \eta'}(m)$ would be finite.
  However, there would then be a parameter value $l$ such that $\eta_l$ is infinitely larger than any $\eta'_{k'}$, when we have $\length{k'} \le m$.
  As $\eta'$ does not have imix, this would violate the assumed $\eta' \quasile_\rmnu \eta$.
\end{proof}

For a further investigation of the structure of parameterizations, we need two more definitions from order theory.

\begin{definition}
  A partially ordered set is a \emph{lattice} if every two elements have a least upper bound and a greatest lower bound.
  A lattice is \emph{bounded} if it has a greatest element and a least element.
\end{definition}

\begin{definition}
  A nonempty up-set $F$ of a lattice is a \emph{filter} if every greatest lower bound of elements of $F$ is also an element of $F$.
  A filter is \emph{principal} if it contains a least element.
\end{definition}

Already without any reference to complexity classes we can characterize the partially ordered sets corresponding to our orders on parameterizations.

\begin{lemma}
\label{lem:lattices}
  The partially ordered sets on equivalence classes of parameterizations obtained from $\quasile_\rmnu$ and $\quasile$ are lattices.
\end{lemma}
\begin{proof}
  We shall show only that the partially ordered sets contain greatest lower bounds.
  The presence of least upper bounds can be shown in an entirely similar fashion and is therefore omitted.
  The same proof works for both $\quasile_\rmnu$ and $\quasile$.
  For simplicity we state the proof for $\quasile$.

  Let $\eta_1 \subseteq \binary^+ \times \Omega_1$ and $\eta_2 \subseteq \binary^+ \times \Omega_2$ be two arbitrary parameterizations.
  It suffices to show that there is a parameterization $\eta$ such that we have $\eta \quasile \eta_1$ and $\eta \quasile \eta_2$, and for every other parameterization $\eta'$ that realizes this we have $\eta' \quasile \eta$.
  From this, the claimed existence of a greatest lower bound of any two equivalence classes of parameterizations follows immediately.

  The parameter space underlying our parameterization $\eta$ will be $\Omega_1 \times \Omega_2$, ordered such that $(k_1, k_2)$ comes before $(k_1', k_2')$ if and only if $k_1$ comes before $k_1'$ in $\Omega_1$ and $k_2$ comes before $k_2'$ in $\Omega_2$.
  Observe that this indeed makes $\Omega$ a directed set and that an effective encoding can be derived from the effective encodings of $\Omega_1$ and $\Omega_2$.
  When $\Omega_1$ and $\Omega_2$ are decidable, $\Omega$ is as well.
  We define our parameterization as
  \begin{equation*}
    \eta \deq \{(x, (k_1, k_2)) \st (x, k_1) \in \eta_1 \relor (x, k_2) \in \eta_2\}.
  \end{equation*}
  Surely, for every $x$ the set $\{(k_1, k_2) \st (x, (k_1, k_2)) \in \eta\}$ is a nonempty up-set of $\Omega_1 \times \Omega_2$ using the order we defined.

  We shall now prove that we have $\eta \quasile \eta_1$.
  For this, fix some $\omega_2 \in \Omega_2$ and suppose we have some $x$ and $n$ such that $\mu_{\eta_1}(x) \le n$ holds.
  Then, there exists a $k_1$ with $\length{k_1} \le n$ and $(x, k_1) \in \eta_1$.
  Moreover, we have $(x, \langle k_1, \omega_2\rangle) \in \eta$ and $\length{\langle k_1, \omega_2\rangle}$ is bounded by a computable function of $n$, depending on the fixed $\omega_2$ and the specifics of the way the encodings of $\Omega_1$ and $\Omega_2$ were combined.
  By the same token, we obtain that $\eta \quasile \eta_2$ holds as well.

  Now suppose we have some other $\eta'$ with which $\eta' \quasile \eta_1$ and $\eta' \quasile \eta_2$ hold.
  We need to prove that we have $\eta' \quasile \eta$, or equivalently that $\gap_{\eta', \eta}$ is upper bounded by a computable function.
  For all $n$, taking $X_n \deq \{x \st \mu_\eta(x) \le n\}$, we have $\gap_{\eta', \eta}(n) = \max\{\mu_{\eta'}(x) \st x \in X_n\}$.
  Since for any reasonable encoding of $\Omega_1 \times \Omega_2$ we get, $X_n \subseteq \{x \st \mu_{\eta_1}(x) \le n\} \cup \{x \st \mu_{\eta_2}(x) \le n\}$ and a maximum cannot increase when taken on a subset, we find $\gap_{\eta', \eta}(n) \le \max\{\gap_{\eta', \eta_1}(n), \gap_{\eta', \eta_2}(n)\}$.
  Because both members of the set on the right hand side are computable as a function of $n$ by assumption, the maximum is computable and we obtain a computable upper bound for $\gap_{\eta', \eta}$.
\end{proof}

The above proof demonstrates the possibility, under our definitions, of combining parameterizations in such a way that the composite parameterization relates naturally to its constituents.
In this case, it was shown that we can construct greatest lower bounds and least upper bounds.
Similar constructions get convoluted when we work solely with the natural numbers as our parameter space, or when we require that points of convergence are unique.
Note that in the above proof, there are no `first' parameter values with which an input is a member of the constructed parameterization.
In order to make the order on the combined parameter space transitive, it was based on a conjunction of orders.
On the other hand, the parameterization itself was based on a disjunction of points of convergence.

In the context of orders on parameterizations, we refer to the equivalence classes of parameterizations when speaking simply of parameterizations.
As we have seen in Lemma~\ref{lem:imix}, parameterizations with imix will remain distinct from those without when employing this convention.
Showcasing this convention, we shall augment Lemma~\ref{lem:lattices} for the nonuniform order.

\begin{lemma}
\label{lem:bounded}
  The lattice of parameterizations obtained from $\quasile_\rmnu$ is a bounded lattice.
\end{lemma}
\begin{proof}
  First, we construct a least element.
  Consider $\eta \deq \binary^+ \times \Omega$, for an arbitrary $\Omega$.
  Let $c$ be the minimum length of a parameter value in $\Omega$.
  For every $\eta'$, we have, for all $n$,
  \begin{equation*}
    \gap_{\eta, \eta'}(n) \le c,
  \end{equation*}
  thus we have $\eta \quasile_\rmnu \eta'$.

  Next, we construct a greatest element.
  Consider $\eta \subset \binary^+ \times \bbN$ defined by $\eta \deq \{(x, n) \st \length{x} \le n\}$.
  For every $\eta'$, the gap with $\eta$ is the maximum of a finite set,
  \begin{equation*}
    \gap_{\eta', \eta}(n) = \max\{\mu_{\eta'}(x) \st \length{x} \le n\}.
  \end{equation*}
  Hence, the gap with $\eta$ is finite and we have $\eta' \quasile_\rmnu \eta$.
\end{proof}

So far, we have looked at parameterizations in isolation.
Turning to complexity theory, we are interested in pairs of a subset of $\binary^+$ and a parameterization.
In its broadest sense, parameterized complexity classes are collections of such pairs.
When a set $A$ is in a parameterized complexity class \cl{C} with parameterization $\eta$, we write $(A, \eta) \in \cl{C}$.
Parameterized complexity classes can be used to define ordered sets of parameterizations.
We shall give names to two of such sets.

\begin{definition}
  Given a partial order $\le$ on parameterizations and a parameterized complexity class \cl{C}, we denote the partially ordered set of parameterizations that may put a set in \cl{C} by
  \begin{align*}
    \calL^\le_\cl{C} &\deq (\{\eta \st \exists A: (A, \eta) \in \cl{C}\}, \le).
  \intertext{For a fixed set $A$, the partially ordered set of parameterizations with which $A$ is in \cl{C} is denoted by}
    \calF^\le_{(A, \cl{C})} &\deq (\{\eta \st (A, \eta) \in \cl{C}\}, \le).
  \end{align*}
\end{definition}

\section{Nonuniform Parameterized Complexity}
\label{sec:nonuniform}

Parameterized complexity classes can be derived from nonparameterized complexity classes \citep{flum2003describing}.
In this section we shall look at a nonuniform way of doing so and study the structure of the sets of parameterizations associated with the resulting parameterized complexity classes.
The complexity class for which we shall define parameterized variants will be \cl{P}, although other classes can be used equally well.

Lifting \cl{P} to different realms of analysis is not new and several ways of doing so have been studied previously.
Using $\exists$ as an operator, the nondeterministic counterpart, \cl{NP}, of \cl{P} has been obtained as \cl{$\exists$P}.
In probabilistic complexity theory, the \cl{\textit{BP}}~operator was derived from the complexity class \cl{\textit{BP}P} by \citet{schoning1989probabilistic} to define various probabilistic complexity classes from their deterministic counterparts.
For defining parameterized complexity classes we shall consider the nonuniform \clXnu{}~operator, of which uniform versions will be studied in the next section.
Note that we indicate nonuniformity with the \clnu{}-subscript.
In general, we shall use unmodified names of parameterized complexity classes for their uniform variants.
In other places (e.g.,~\url{https://complexityzoo.uwaterloo.ca/Complexity_Zoo:X}), special notation is used instead to denote the uniform classes.

\begin{definition}
  A set $A$ is in \emph{\clXnu{P}} with parameterization $\eta$ if for every parameter value $k$ the set $\eta_k$ is a \cl{P}-slice for $A$.
\end{definition}

In the current century, the class \clX{P} has been called \emph{slicewise~\cl{P}}, giving a name to the \clX{}~operator \citep{flum2003describing}.
However, in earlier work, \citet{downey1999parameterized} used ``slicewise \cl{P}'' to denote a different complexity class (\cl{FPT}).

By definition of \cl{P}, a set is a \cl{P}-slice if there is a constant $c$ such that it is an \clO{n^c}-slice.
Because the exponent $c$ is in general unbounded for members of \clXnu{P}, the class \clXnu{P} is very permissive.
A bound on the exponent is obtained by changing the order of the quantifiers in the definition above.

\begin{definition}
  A set $A$ is in \emph{\clnu{FPT}} with parameterization $\eta$ if there is a constant $c$ such that for every parameter value $k$ the set $\eta_k$ is an \clO{n^c}-slice for $A$.
\end{definition}

Sets in \cl{FPT} are said to be \emph{fixed-parameter tractable} as they can be decided in polynomial time with a fixed exponent.
The algorithm and precise running time for the slices of a fixed-parameter tractable set vary depending on the parameter.
In our nonuniform variant \clnu{FPT}, this dependence is unrestricted.

For every value of $c$, we can replace \cl{P} by \clO{n^c} in the definition of \clXnu{P} to get a definition of \clXnu{\clO{n^c}}.
With this definition in place, we get the equalities
\label{eq:xpfpt}
\begin{align*}
  \clXnu{P} &= \clXnu{$\Big(\bigcup_c\clO{n^c}\Big)$}
\shortintertext{and}
  \clnu{FPT} &= \bigcup_c \Big(\clXnu{\clO{n^c}}\Big),
\end{align*}
further illustrating the reversal of the order of dependence on the parameter and the exponent in the definitions of \clXnu{P} and \clnu{FPT}.

The sets of parameterizations relevant to \clXnu{P} and \clnu{FPT} showcase an algebraic structure  similar to that of the set of all parameterizations ordered by the nonuniform order.
\begin{lemma}
\label{lem:nulattices}
  $\calL^{\quasile_\rmnu}_\clXnu{P}$ and $\calL^{\quasile_\rmnu}_\clnu{FPT}$ are bounded lattices.
\end{lemma}
\begin{proof}
  We shall first show that the partially ordered sets are lattices.
  In order to do so, we shall only prove that the partially ordered sets contain greatest lower bounds.
  The presence of least upper bounds can be shown in an entirely similar fashion and is therefore omitted.

  Let \cl{C} be either \clXnu{P} or \clnu{FPT} and let $\eta_1$ and $\eta_2$ be parameterizations in $\calL^{\quasile_\rmnu}_\cl{C}$.
  The parameterization $\eta$ as constructed in the proof of Lemma~\ref{lem:lattices} is a greatest lower bound in the lattice of all parameterizations.
  Since the union of a set in \clO{n^{c_1}} and a set in \clO{n^{c_2}} is in \clO{n^{\max\{c_1, c_2\}}}, the parameterization $\eta$ puts any trivial set, such as the empty set, in \cl{C} and thus $\eta$ is contained in $\calL^{\quasile_\rmnu}_\cl{C}$.
  Hence $\calL^{\quasile_\rmnu}_\cl{C}$ indeed contains all greatest lower bounds.

  Next, we shall show that the lattices contain a least and a greatest element.
  Sets in \cl{P} are put in both \clXnu{P} and \clnu{FPT} by the parameterization $\binary^+ \times \Omega$, where $\Omega$ is an arbitrary parameter space.
  In the proof of Lemma~\ref{lem:bounded} we have seen that this parameterization is below any other.
  Hence it is a least element in both $\calL^{\quasile_\rmnu}_\clXnu{P}$ and $\calL^{\quasile_\rmnu}_\clnu{FPT}$.
  The greatest element of the lattice of parameterizations, as constructed in the proof of Lemma~\ref{lem:bounded}, is present in both $\calL^{\quasile_\rmnu}_\clXnu{P}$ and $\calL^{\quasile_\rmnu}_\clnu{FPT}$ and thus a greatest element for both these lattices.
\end{proof}

Focussing on \clXnu{P} for a moment, we find that the set of parameterizations with which any particular set is put in \clXnu{P} sits nicely inside $\calL^{\quasile_\rmnu}_\clXnu{P}$.
\begin{lemma}
\label{lem:nuxpfilter}
  For every set $A$, $\calF^{\quasile_\rmnu}_{(A, \clXnu{P})}$ is a filter in $\calL^{\quasile_\rmnu}_\clXnu{P}$.
\end{lemma}
\begin{proof}
  Let $A$ be an arbitrary set.
  We need to prove that $\calF^{\quasile_\rmnu}_{(A, \clXnu{P})}$ is a nonempty up-set of $\calL^{\quasile_\rmnu}_\clXnu{P}$ that includes all greatest lower bounds of its members.

  As $\calF^{\quasile_\rmnu}_{(A, \clXnu{P})}$ includes all parameterizations $\eta$ for which, for every parameter value $k$, the set $\eta_k$ is finite, it is nonempty.
  To see that it is an up-set, first observe that for every parameterization $\eta$ in $\calL^{\quasile_\rmnu}_\clXnu{P}$ and every parameter value $k$ in the corresponding parameter space, the set $\eta_k$ is a \cl{P}-slice of some set and hence $\eta_k$ is in \cl{P}.
  Next, note that whenever there is a parameterization $\eta'$ in $\calF^{\quasile_\rmnu}_{(A, \clXnu{P})}$ that is below a parameterization $\eta$, then for every parameter value $k$ in the parameter space of $\eta$ there is a parameter value $k'$ in the parameter space of $\eta'$ such that we have $\eta_k \subseteq \eta'_{k'}$.
  Because every subset of a \cl{P}-slice that is in \cl{P} is itself a \cl{P}-slice, $\eta_k$ is a \cl{P}-slice for $A$.
  Hence, as the parameter values were chosen arbitrarily, we may conclude that $\eta$ is in $\calF^{\quasile_\rmnu}_{(A, \clXnu{P})}$ and thus that $\calF^{\quasile_\rmnu}_{(A, \clXnu{P})}$ is an up-set of $\calL^{\quasile_\rmnu}_\clXnu{P}$.

  Let $\eta_1$ and $\eta_2$ be two parameterizations in $\calF^{\quasile_\rmnu}_{(A, \clXnu{P})}$.
  To prove that $\calF^{\quasile_\rmnu}_{(A, \clXnu{P})}$ includes all greatest lower bounds of its members, it suffices to show that the greatest lower bound $\eta$ as constructed in the proof of Lemma~\ref{lem:lattices} is in $\calF^{\quasile_\rmnu}_{(A, \clXnu{P})}$.
  That this is the case readily follows from the ability to combine any two \cl{P}-approximations for $A$ into a \cl{P}-approximation for $A$ of which the domain is the union of the domains of its constituents.
\end{proof}

The orders on parameterizations can be thought of as an inverse ranking of how powerful parameterizations are.
When a parameterization is below another, it signifies that the convergence behavior of this parameterization is an improvement over that of the other.
This improvement is of a much stronger kind than the improvements made in parameterized algorithms races, where improvements are sought within a parameterization \citep{komusiewicz2012new,fellows2013towards}.
In this regard, the best parameterizations in a filter of parameterizations such as that of the previous lemma are those that are below most others.
Hence a set $A$ admits an optimal parameterization with respect to \clXnu{P} if the filter $\calF^{\quasile_\rmnu}_{(A, \clXnu{P})}$ is principal.
\begin{theorem}
\label{thm:nuxpprincipal}
  For any set $A$, the filter $\calF^{\quasile_\rmnu}_{(A, \clXnu{P})}$ is principal.
\end{theorem}
\begin{proof}
  Let $S_1, S_2, \ldots$ be an enumeration of the \cl{P}-slices for $A$.
  Consider the parameterization $\eta \subseteq \binary^+ \times \bbN$ given by
  \begin{equation*}
    \eta_k \deq \bigcup_{i \le k} S_i.
  \end{equation*}
  By definition $A$ is in \clXnu{P} with $\eta$ and by construction $\eta$ is a least element in $\calF^{\quasile_\rmnu}_{(A, \clXnu{P})}$.
\end{proof}

We shall call a least element in the filter corresponding to some set a \emph{principal parameterization} for that set.
While Theorem~\ref{thm:nuxpprincipal} shows that all sets have principal parameterizations with respect to \clXnu{P}, this is not a given for arbitrary parameterized complexity classes.
When they exist, principal parameterizations provide insight into some of the computational complexity of a set.
Indeed, as a consequence of Theorem~\ref{thm:nuxpprincipal} there is a one-to-one correspondence between the imix property of a principal parameterization with respect to \clXnu{P} and the levelability of a set.

\begin{corollary}
\label{cor:nuxpimix}
  A set is \levelable{\cl{P}} (\immune{\cl{P}}) if and only if a principal parameterization for the induced nonuniform filter with respect to \clXnu{P} has (does not have) imix.
\end{corollary}

Note that the filter induced by a \cl{P}-bi-immune set consists of a single equivalence class of parameterizations, namely that of parameterizations $\eta$ where for every parameter value $k$ the set $\eta_k$ is finite.

In case the existence of a principal parameterization is not a given, a statement like Corollary~\ref{cor:nuxpimix} is still valid when there is a parameterization such that it and all parameterizations below it have imix.
This motivates the following definition.

\begin{definition}
  A set $A$ is \emph{\levelable{\clXnu{P}}} (\emph{\immune{\clXnu{P}}}) if there is a parameterization $\eta$ with which $A$ is in \clXnu{P} such that every parameterization $\eta' \quasile_\rmnu \eta$ with which $A$ is in \clXnu{P} has (does not have) imix.
\end{definition}

In these definitions, \cl{P} can be replaced by other complexity classes, in particular, for any $c$, by \clO{n^c}.
However, by the nature of the definition of \clnu{FPT} we should not blindly replace \clXnu{P} by \clnu{FPT} in these definitions, but instead include a dependence on the exponent $c$.

\begin{definition}
  A set is \emph{\levelable{\clnu{FPT}}} (\emph{\immune{\clnu{FPT}}}) if there is a constant $c_0$ such that for all $c \ge c_0$ the set is \levelable{\clXnu{\clO{n^c}}} (\immune{\clXnu{\clO{n^c}}}).
\end{definition}

It is an open problem whether this definition is different from the alternative obtained by replacing \clXnu{P} by \clnu{FPT} in our initial definition of parameterized levelability.
If the two notions are the same, then for every \levelable{\clnu{FPT}} set $A$ there is a constant $c$ such that every \cl{P}-slice for $A$ can be infinitely extended to an \clO{n^c}-slice for $A$.
We do not expect every set we wish to call \levelable{\clnu{FPT}} to show this behavior and shall go with the tailored definition of levelability for \clnu{FPT}.
This definition works the way we want it to and we feel that keeping the exponent fixed in the analysis reflects the spirit of fixed-parameter tractability.

Our definitions are so that from Theorem~\ref{thm:nuxpprincipal} and Corollary~\ref{cor:nuxpimix} we get an identification of levelability with respect to \cl{P} and levelability with respect to \clXnu{P}.
\begin{corollary}
\label{cor:xplevelable}
  A set is \levelable{\cl{P}} (\immune{\cl{P}}) if and only if it is \levelable{\clXnu{P}} (\immune{\clXnu{P}}).
\end{corollary}

The case for levelability with respect to \clnu{FPT} is more complicated.
As with \clXnu{P}, we find that the set of parameterizations with which any particular set is put in \clnu{FPT} is a filter.
\begin{lemma}
\label{lem:nufptfilter}
  For every set $A$, $\calF^{\quasile_\rmnu}_{(A, \clnu{FPT})}$ is a filter in $\calL^{\quasile_\rmnu}_\clnu{FPT}$.
\end{lemma}
\begin{proof}
  This lemma can be proven similarly to Lemma~\ref{lem:nuxpfilter}.
  For the current lemma, however, we need to keep track of the exponent in the running time of the \cl{P}-approximations involved in the proof.
  In general, when multiple \cl{P}-approximations are at play, it is possible to compute all of them within a polynomial running time of which the exponent is the maximum of the exponents of the individual polynomial running times.
\end{proof}

The existence of optimal parameterizations is of interest with respect to \clnu{FPT} too.
Contrary to the case for \clXnu{P}, not every filter of parameterizations with which a set is put in \clnu{FPT} is principal.
For some sets, however, principality of the corresponding filter is easily shown.
\begin{theorem}
\label{thm:nufptprincipal}
  For any set $A$ that is \immune{\cl{P}}, the filter $\calF^{\quasile_\rmnu}_{(A, \clnu{FPT})}$ is principal.
\end{theorem}
\begin{proof}
  By definition of being \immune{\cl{P}}, $A$ has a maximal \cl{P}-slice $S$.
  For some constant $c$, this \cl{P}-slice $S$ is also an \clO{n^c}-slice.
  Any parameterization $\eta$ with which $A$ is in \clnu{FPT} and for which there is a parameter value $k$ such that $\eta_k$ equals $S$ is equivalent to any parameterization below it, hence such an $\eta$ is principal.
\end{proof}

Of course, a principal parameterization with respect to \clnu{FPT} for any \immune{\cl{P}} set does not have imix and it follows that such sets are also \immune{\clnu{FPT}}.
\begin{corollary}
\label{cor:nufptimmune}
  If a set is \immune{\cl{P}}, it is \immune{\clnu{FPT}}.
\end{corollary}

The proof of Theorem~\ref{thm:nufptprincipal} inspires alternative characterizations of the \immune{\clnu{FPT}} and \levelable{\clnu{FPT}} sets.
\begin{lemma}
\label{lem:nufptnc}
  A set is \levelable{\clnu{FPT}} (\immune{\clnu{FPT}}) if and only if there is a constant $c_0$ such that for all $c \ge c_0$ the set is \levelable{\clO{n^c}} (\immune{\clO{n^c}}).
\end{lemma}
\begin{proof}
  For any $c$, we can replace \cl{P} by \clO{n^c} in Theorem~\ref{thm:nuxpprincipal}.
  The current lemma then follows from the associated variants of Corollary~\ref{cor:xplevelable}.
\end{proof}

While we have seen that levelability with respect to \clXnu{P} is identical to levelability with respect to \cl{P}, the situation with respect to \clnu{FPT} is different.
We shall provide a constructive proof of the fact that the converse of Corollary~\ref{cor:nufptimmune} does not hold.
The set we are about to construct has the remarkable property that it is \immune{$\bigO(n^c)$} for infinitely many values of $c$, yet each of the bi-immune parts is still decidable in polynomial time.
Conceptually, we show that it is possible to diagonalize against polynomial time machines in polynomial time.
\begin{theorem}
\label{thm:nonnufptlevelable}
  There are \levelable{\cl{P}} sets that are not \levelable{\clnu{FPT}}.
\end{theorem}
\begin{proof}
  We shall prove the theorem by constructing a \levelable{\cl{P}} set $A$ for which there are infinitely many $c \in \bbN$ such that $A$ is \immune{\clO{n^c}}.
  For this, let $\phi_1, \phi_2, \ldots$ be an effective enumeration of all partial decision procedures and denote by $\phi^{\downarrow c}_i$ the $i$th partial decision procedure restricted to running time $n^c$ so that for all $i$ and all $c \le d$ we have $\dom \phi^{\downarrow c}_i \subseteq \phi^{\downarrow d}_i$.
  Additionally, let $\langle c, x\rangle$ be the result of a pairing function such as Cantor's.

  Consider the following recursive procedure for deciding a set $A$.
  On input $\langle c, x\rangle$, the procedure takes the following steps.
  \begin{enumerate}
  \item
    \hspace{\stretch{1}}\textit{Determine a set $I$ of procedures consistent with an initial segment of $A$:}
    \begin{enumerate}
    \item
      Set $I \deq \{1, 2, \ldots, \length{\langle c, x\rangle}\}$.
    \item
      For $d$ in $\{1, 2, \ldots c\}$ and $y$ in $\{1, 2, \ldots, \length{\langle c, x\rangle}\}$ do:
      \\\-\quad Set $I \deq \{i \st i \in I \reland \phi^{\downarrow 2c}_i(\langle d, y\rangle) = A(\langle d, y\rangle)\}$.
    \end{enumerate}
  \item
    \hspace{\stretch{1}}\textit{Try to make a procedure from $I$ inconsistent with $A$:}
    \begin{enumerate}
    \item
      For $i$ in $I$ do:
      \\\-\quad If $\phi^{\downarrow 2c}_i(\langle c, x\rangle) \in \{0, 1\}$ then return $1 - \phi^{\downarrow 2c}_i(\langle c, x\rangle)$.
    \item
      return $0$. \hspace{\stretch{1}}\textit{(arbitrary)}
    \end{enumerate}
  \end{enumerate}
  The first stage of this procedure performs at most $c\length{\langle c,+x\rangle}^2$ simulations of computations, each with a running time of $n^{2c}$ where the input length $n$ is at most the length of $\langle c, \length{\langle c, x\rangle}\rangle$.
  Besides these simulations, this procedure needs access to a segment of $A$ to test against.
  This segment can be computed recursively, with a recursion depth bounded by the iterated logarithm of $x$.
  By using dynamic programming the time required to compute the segment is insignificant with respect to the total running time of the entire procedure.

  The second stage of the procedure requires the simulation of at most $\length{\langle c, x\rangle}$ computations, each with a running time of $n^{2c}$, where we have $n = \length{\langle c, x\rangle}$.
  Using efficient simulation \citep{arora2009computational}, this puts the running time of the second stage in $\bigO(n^{2c + 2})$.
  Note that the running time of the procedure is thus not polynomial in $\length{\langle c, x\rangle}$, as $c$ appears in the exponent and is not independent of the input.

  For any fixed $c$, the set $\{\langle d, x\rangle \st d \le c \reland x \in \binary^+\}$ is a \clO{n^{2c + 2}}-slice of $A$.
  It is not a maximal \cl{P}-slice, as for larger values of $c$ infinitely many elements are introduced in the corresponding sets.
  However, we claim that it is a maximal \clO{n^{2c + 2}}-slice and thus that $A$ is \levelable{\cl{P}} yet not \levelable{\clnu{FPT}}.
  Suppose towards a contradiction that there is an infinite \clO{n^{2c + 2}} slice $S \deq \{\langle d, x\rangle \st d > c \reland x \in \binary^+\}$ for $A$ that is the domain of a partial decision procedure that occurs in our enumeration with index $i$.
  When $\length{\langle d, x\rangle}$ outgrows $i$ in the first stage of our procedure for deciding $A$ will include $i$ in $I$.
  Because there are only finitely many values for $d$ and $x$ such that $\length{\langle d, x\rangle}$ is smaller than $i$, we may assume that $I$ contains $i$ at the start of the second stage of our procedure.
  Now for every element of $S$, either the second stage invalidates $i$ as the index of a \clO{n^{2c + 2}}-slice for $A$, or an index smaller than $i$ is removed from $I$ for all subsequent values of $d$ and $x$.
  The latter of these possibilities can happen at most $i$ times, so, since we assumed that $S$ is infinite, eventually $i$ must be invalidated.
  This contradicts our assumption that $i$ was the index of a partial decision procedure for $S$ with a running time in $\bigO(n^{2c + 2})$.
  Note that our time bounds were chosen as they are so that already for $d = c + 1$ we find $2d \ge 2c + 2$.
\end{proof}

Note that although the set constructed in the above proof is not \levelable{\clnu{FPT}} this does not mean that it is necessarily \immune{\clnu{FPT}}.
It could be that the set is neither \levelable{\clnu{FPT}} nor \immune{\clnu{FPT}}.
We shall now turn to the existence of \levelable{\clnu{FPT}} sets.
Our notion of a \emph{reduction} is that of a membership preserving polynomial time computable funcion, in other words, that of a Karp~reduction.
\begin{theorem}
\label{thm:nufptlevelable}
  Every set outside \cl{P} from which there is a \lli{} reduction to itself is \levelable{\clnu{FPT}}.
\end{theorem}
\begin{proof}
  Let $A$ be a set outside \cl{P} and $f$ a \lli{} reduction from $A$ to itself.
  Suppose that $A$ is not \levelable{\clnu{FPT}} and has a maximal \clO{n^c}-slice $S$, with $c$ high enough for $f$ to be computable in time $\bigO(n^{c - 1})$.
  The sets
  \begin{align*}
    S' &\deq \{x \st x \notin S \reland f(x) \in S\}, \\
    S_x &\deq \{x, f(x), f(f(x)), \ldots\}
  \end{align*}
  are, by nature of $f$, also \clO{n^c}-slices for $A$.
  For $S'$ this requires the linear bound on the length of the output of $f$, where for $S_x$ this requires that $f$ is length increasing.
  Furthermore, $S'$ satisfies $S \cap S' = \emptyset$.

  By the assumed maximality of $S$, for every $x$ there are only finitely many elements in the set $S \backslash S_x$.
  However, since $A$ is not in \cl{P}, there are infinitely many $x$ outside $S$ and for each of these the set $S_x$ intersects $S'$.
  Hence $S'$ is infinite, contradicting the maximality of $S$.
\end{proof}

The existence of \levelable{\clnu{FPT}} sets now follows from the existence of sets outside \cl{P} that have a \lli{} reduction to itself.
\begin{lemma}
\label{lem:llireduction}
  There are sets outside \cl{P} that have a \lli{} reduction to itself.
\end{lemma}
\begin{proof}
  Let $X$ be a set outside \cl{P} and consider its cylindrification
  \begin{equation*}
    A \deq \{\langle x, y\rangle \st x \in X, y \in \bbN\}.
  \end{equation*}
  Note that $A$ too is not in \cl{P}.
  For most reasonable choices of pairing functions, the function $f$ defined by
  \begin{equation*}
    f(\langle x, y\rangle) \deq \langle x, 2 \cdot y\rangle
  \end{equation*}
  is a \lli{} reduction from $A$ to itself.
\end{proof}

It should be noted that certain padding functions give rise to \lli{} reductions.
There are padding functions, $\pad$, for which there exists a constant $c$ such that for all $x, y$ we have
\begin{equation*}
  \frac{1}{c} \cdot (\length{x} + \length{y}) \le \length{\pad(x, y)} \le c \cdot (\length{x} + \length{y}).
\end{equation*}
Padding functions that meet this enhanced honesty criterion can be turned into \lli{} reductions by mapping $x$ to $\pad(x, 0^{c \cdot \length{x}})$.

It was shown~\citep{young1983some} that every honest, one-one, polynomially computable function is the productive function for the complement of some $k$-creative set.
In particular this means that there are $k$-creative sets where the influence of the associated productive function on the length of its input is linear.
Such functions can again be turned into \lli{} reductions and hence, assuming $\cl{P} \ne \cl{NP}$, there are also $k$-creative sets that are \levelable{\clnu{FPT}}.
%TODO: In a later paper, relate \levelable{\cl{FPT}} to (parameterized) self-reducibility

Starting from Lemma~\ref{lem:nufptnc} we can augment Theorem~\ref{thm:nufptprincipal} and thus answer the principality question for all sets that are not \levelable{\clnu{FPT}}.
\begin{theorem}
\label{thm:nufptnonprincipal}
  For any set $A$ that is \levelable{\cl{P}} and not \levelable{\clnu{FPT}}, the filter $\calF^{\quasile_\rmnu}_{(A, \clnu{FPT})}$ is nonprincipal.
\end{theorem}
\begin{proof}
  In case $A$ is not \levelable{\clnu{FPT}}, there is an unbounded set of values $c$ such that $A$ has a maximal \clO{n^c}-slice.
  When such a set $A$ is \levelable{\cl{P}}, then for every \clO{n^c}-slice $S$ for $A$ that is maximal, there is a constant $d$ such that $A$ has a maximal \clO{n^d}-slice that is infinitely larger than $S$.
  Hence $\calF^{\quasile_\rmnu}_{(A, \clnu{FPT})}$ cannot be principal.
\end{proof}

For \immune{\clnu{FPT}} sets, the principality of their filters with respect to \clnu{FPT} is still an open problem.
Next to identifying sets that admit optimal parameterizations, we can study to what degree a filter of parameterizations is unique to a set.
The filter of parameterizations with respect to a parameterized complexity class is then taken as a representation of the distribution of complexity inside a set.
This approach is a continuation of an idea by \citet{orponen1986classification} who represented the complexity characteristics of a set by the filter of its complexity cores.
Where this idea was shown fruitless when using proper cores, our parameterized setting is promising.
For sets $A, B$, let $A \symdiff B$ denote the symmetric difference $(A \backslash B) \cup (B \backslash A)$.
\begin{theorem}
\label{thm:nufptsymdiffeq}
  For any set $X$ in \cl{P} and any set $A$ we have
  \begin{equation*}
    \calF^{\quasile_\rmnu}_{(A, \clnu{FPT})} = \calF^{\quasile_\rmnu}_{(A \symdiff X, \clnu{FPT})}.
  \end{equation*}
\end{theorem}
\begin{proof}
  This follows from the fact that given $X$, any slice for $A$ can be turned into a slice for $A \symdiff X$ and vice versa.
\end{proof}

Intuitively, the above theorem states that taking the symmetric difference with an easy set does not alter the distribution of complexity.
Similarly, we find that the symmetric difference of two sets with the same distribution of complexity is easier than either of the initial sets.
\begin{theorem}
\label{thm:nufptsymdiffsubeq}
  For any two sets $A, B$ satisfying $\calF^{\quasile_\rmnu}_{(A, \clnu{FPT})} = \calF^{\quasile_\rmnu}_{(B, \clnu{FPT})}$ we have
  \begin{equation*}
    \calF^{\quasile_\rmnu}_{(A, \clnu{FPT})} \subseteq \calF^{\quasile_\rmnu}_{(A \symdiff B, \clnu{FPT})}.
  \end{equation*}
\end{theorem}
\begin{proof}
  This follows from the fact that any slice for $A$ is also a slice for $B$ and can thus be turned into a slice for $A \symdiff B$.
\end{proof}

Even though Theorem~\ref{thm:nufptsymdiffsubeq} asserts that the symmetric difference of two sets that share all their parameterizations is easier than either of the sets, it does not guarantee that this symmetric difference is in \cl{P}.
If this would be the case, a filter with respect to \clnu{FPT} would uniquely define a set up to variations in \cl{P}.
Of comparable flavor is the Berman--Hartmanis conjecture \citep{young1983some,li1997introduction}, which asserts that completeness for \cl{NP} uniquely defines a set up to isomorphisms computable in polynomial time.
We conjecture that indeed the filter of the symmetric difference collapses to that of a set in \cl{P}.
In other words, we conjecture that for all $A, B$, whenever we have $\calF^{\quasile_\rmnu}_{(A, \clnu{FPT})} = \calF^{\quasile_\rmnu}_{(B, \clnu{FPT})}$ there is some $X$ in \cl{P} such that we have $B = A \symdiff X$.
Here, we should keep in mind that taking the symmetric difference with some set is an involution and we have $A \symdiff (A \symdiff X) = X$.
Because all sets in \cl{P} have the same associated filter, our conjecture can be expressed elegantly as follows.
\begin{conjecture}
\label{con:nusymdiff}
  For any two sets $A, B$ we have
  \begin{equation*}
    \calF^{\quasile_\rmnu}_{(A, \clnu{FPT})} = \calF^{\quasile_\rmnu}_{(B, \clnu{FPT})} \iff \calF^{\quasile_\rmnu}_{(A \symdiff B, \clnu{FPT})} = \calF^{\quasile_\rmnu}_{(\emptyset, \clnu{FPT})}.
  \end{equation*}
\end{conjecture}
If true, a separation result would follow.
Namely, for any two sets $A, B$ of which the symmetric difference is outside \cl{P}, there would exist a set in \cl{P} separating $A$ and $B$ in the sense that the separating set contains a \cl{P}-core for precisely one of the two.

\section{Uniform Parameterized Complexity}
\label{sec:uniform}

The classes \clXnu{P} and \clnu{FPT} were nonuniform in two ways.
Firstly, we did not require the parameter dependence of the running time of the \cl{P}-approximations of slices to have a uniformly computable upper bound as a function of the parameter value to which the slices belong.
Put differently, in the equivalent characterizations of the complexity classes on page~\pageref{eq:xpfpt}, we observe that the \clXnu{} operator is used on (a union of) \clO{n^c} complexity classes without restrictions to the constant hidden in the $\bigO$-notation.
Secondly, we did not require the existence of a procedure to instantiate the \cl{P}-approximations of the slices from the parameter values to which the slices belong.

For the majority of this section, we shall consider two additional complexity classes for fixed-parameter tractability, both exhibiting more uniformity than \clnu{FPT}.
The first of them is a semi-uniform variant, which is known as \emph{uniform \cl{FPT}} by \citet{downey1999parameterized}.
In this variant the \cl{P}-approximations are required to form a uniform collection.
\begin{definition}
  A set $A$ is in \clsu{FPT} with parameterization $\eta$ if there is a constant $c$ and a Turing machine $\Psi$ taking two inputs, such that for every parameter value $k$ the partial application of $\Psi$ to $k$ yields an \clO{n^c}-approximation, $\Psi(k, \cdot)$, for $A$ with domain $\eta_k$.
\end{definition}
The second of our classes is a fully uniform variant, for which we use no subscript.
In \citep{downey1999parameterized}, this class is known as \emph{strongly uniform \cl{FPT}}.
\begin{definition}
  A set $A$ is in \cl{FPT} with parameterization $\eta$ if it is in \clsu{FPT} with constant $c$ and Turing machine $\Psi$ and there is a computable function $f$ such that the partial application of $\Psi$ to $k$ has a running time bounded by $f(k) \cdot n^c$.
\end{definition}

The parameterizations that can put sets in \clsu{FPT} on \cl{FPT} can be identified.
This will help us characterize the structure of $\calL^{\quasile_\rmnu}_\clsu{FPT}$ and $\calL^\quasile_\cl{FPT}$.
\begin{lemma}
\label{lem:fptdecidable}
  A parameterization $\eta$ is in $\calL^{\quasile_\rmnu}_\clsu{FPT}$ if and only if there is a constant $c$ such that, uniformly in $k$, membership of any $x$ in $\eta_k$ can be decided by a decision procedure with a running time in $\bigO(n^c)$.

  Moreover, $\eta$ is in $\calL^\quasile_\cl{FPT}$ if and only if there is additionally a computable function $f$ such that membership of any $(x, k)$ in $\eta$ can be decided by a decision procedure with a running time bounded by $f(k) \cdot \length{x}^c$.
\end{lemma}
\begin{proof}
  When a set is in \clsu{FPT} or \cl{FPT} with some parameterization, the Turing machine witnessing such a membership can easily be modified to decide membership of the parameterization.
  Conversely, a trivial set such as the empty set is put in \clsu{FPT} or \cl{FPT} by every parameterization that meets the respective running time bound.
\end{proof}

Thus the parameterizations relevant for semi-uniform and strongly uniform fixed-parameter tractability are decidable in a time bound that is polynomial in the length of their first component.
This makes that convergence of computation with a given parameter value becomes less of a \emph{promise} for the input, as it was for our nonuniform classes, and more of a \emph{property} of the input.
Moreover, it follows from the previous lemma that the minimization function $\mu_\eta$ of a parameterization $\eta$ taken from $\calL^{\quasile_\rmnu}_\clsu{FPT}$ or $\calL^\quasile_\cl{FPT}$ is computable.
Having identified the parameterizations relevant for semi-uniform and strongly uniform fixed-parameter tractability, we explore the algebraic structure they form.
\begin{lemma}
  $\calL^{\quasile_\rmnu}_\clsu{FPT}$ and $\calL^\quasile_\cl{FPT}$ are bounded lattices.
\end{lemma}
\begin{proof}
  We shall phrase the proof as a proof for $\calL^\quasile_\cl{FPT}$.
  Only the first part is relevant for $\calL^{\quasile_\rmnu}_\clsu{FPT}$ as the rest of the proof for this partially ordered set is identical to Lemma~\ref{lem:nulattices}.

  That $\calL^\quasile_\cl{FPT}$ is a lattice follows from Lemma~\ref{lem:fptdecidable}.
  The greatest lower bound of two decidable parameterizations as constructed in the proof of Lemma~\ref{lem:lattices} is again decidable, hence $\calL^\quasile_\cl{FPT}$ contains greatest lower bounds.
  As before, the remainder of proving that $\calL^\quasile_\cl{FPT}$ is a lattice is now routine.

  Similar to Lemma~\ref{lem:nulattices}, we observe that $\binary^+ \times \Omega$ is a least element in $\calL^\quasile_\cl{FPT}$ for an arbitrary parameter space $\Omega$.
  The case for a greatest element is more subtle.
  In $\calL^\quasile_\cl{FPT}$, a parameterization is below another not simply when there is a finite bound on the gap function of the two, but only when there is such a bound that is computable.
  For this reason, given the parameterization $\eta \deq \{(x, k) \st \length{x} \le k\}$ of Lemma~\ref{lem:bounded}, the argument that for all slices in $\eta$ are finite is not sufficient to show that $\eta$ is a greatest element.
  Surely, it follows from Lemma~\ref{lem:fptdecidable} that $\eta$ is in $\calL^\quasile_\cl{FPT}$.
  That it is indeed also a greatest element requires that, uniformly in $k$, we can compute the number of elements in $\eta_k$.
  Because of this, the maximum
  \begin{equation*}
    \gap_{\eta', \eta}(n) = \max\{\mu_{\eta'}(x) \st \length{x} \le n\}
  \end{equation*}
  is computable and $\eta$ is a greatest element in $\calL^\quasile_\cl{FPT}$.
\end{proof}

While all parameterizations in $\calL^{\quasile_\rmnu}_\clsu{FPT}$ of which all slices are finite are $\quasile_\rmnu$-equivalent, not all parameterizations in $\calL^\quasile_\cl{FPT}$ of which all slices are finite are $\quasile$-equivalent.
In particular, the number of elements in a finite set in \cl{P} need not be computable.
It is for this reason that the existence of a greatest element in $\calL^\quasile_\cl{FPT}$ requires a more delicate proof than the existence of a greatest element in $\calL^{\quasile_\rmnu}_\clsu{FPT}$.

We have seen that the minimization function for parameterizations that put a set in \clsu{FPT} or \cl{FPT} are computable.
As a result of this, we find that the sets that can be put in \clsu{FPT} or \cl{FPT} are all decidable.
Indeed, the sets that can be put in \clsu{FPT} or \cl{FPT} are \emph{precisely} the decidable sets.
\begin{lemma}
  For every decidable set $A$, $\calF^{\quasile_\rmnu}_{(A, \clsu{FPT})}$ is a filter in $\calL^{\quasile_\rmnu}_\clsu{FPT}$ and $\calF^\quasile_{(A, \cl{FPT})}$ is a filter in $\calL^\quasile_\cl{FPT}$.
\end{lemma}
\begin{proof}
  Let $A$ be any decidable set.
  We shall first prove that $\calF^{\quasile_\rmnu}_{(A, \clsu{FPT})}$ and $\calF^\quasile_{(A, \cl{FPT})}$ are nonempty.
  For this, let $\Phi$ be a decision procedure for $A$ and consider the parameterization
  \begin{equation*}
    \eta \deq \{(x, k) \st \text{$\Phi$ decides membership of $x$ in at most $k$ steps}\}.
  \end{equation*}
  Immediately, we see that $\eta$ is in both $\calF^{\quasile_\rmnu}_{(A, \clsu{FPT})}$ and $\calF^\quasile_{(A, \cl{FPT})}$.
  The running times of the corresponding approximations for $A$ do not depend on the length of the input at all.

  Next, we shall prove that $\calF^\quasile_{(A, \cl{FPT})}$ is an up-set of $\calL^\quasile_\cl{FPT}$.
  The case for $\calF^{\quasile_\rmnu}_{(A, \clsu{FPT})}$ is simpler.
  Suppose we have a parameterization $\eta$ in $\calL^\quasile_\cl{FPT}$ and a parameterization $\eta' \subseteq \binary^+ \times \Omega'$ in $\calF^\quasile_{(A, \cl{FPT})}$ that is below $\eta$.
  We need to show that $A$ is in \cl{FPT} with parameterization $\eta$ as well.
  Let $\Psi$ be the Turing machine witnessing that $A$ is in \cl{FPT} with $\eta'$, and consider the following procedure on input $(x, k)$.
  \begin{enumerate}
  \item
    If $(x, k)$ is not in $\eta$, return nothing.
  \item
  \label{alg:upset:decide}
    For $\omega'$ in (an unbounded subset of) $\Omega'$ do:
    \\\-\quad If $\Psi(x, \omega') \in \{0, 1\}$ then return $\Psi(x, \omega')$.
  \end{enumerate}
  For precisely those $(x, k)$ that are in $\eta$, this procedure decides membership of $x$ in $A$.
  The time it needs to do so depends on where in the enumeration of (an unbounded subset of) $\Omega'$ we encounter an $\omega'$ such that $(x, \omega')$ is in $\eta'$.
  For each of the values in the enumeration up to and including this $\omega'$, the time needed for step~\ref{alg:upset:decide} is, for some $c$ depending on $\Psi$, in $\bigO(n^c)$ with a hidden constant depending on the parameter value.
  In the case of \cl{FPT} the dependence on the parameter value is upper bounded by some computable function.
  Because we are considering the uniform order on parameterizations, $\quasile$, it is furthermore possible to compute an upper bound on the set of parameter values considered in step~\ref{alg:upset:decide} of the procedure.
  Therefore, the running time of the above procedure is so that it witnesses that $A$ is in \cl{FPT} with parameterization $\eta$ and we have shown that $\calF^\quasile_{(A, \cl{FPT})}$ is an up-set of $\calL^\quasile_\cl{FPT}$.

  It remains to show that $\calF^{\quasile_\rmnu}_{(A, \clsu{FPT})}$ and $\calF^\quasile_{(A, \cl{FPT})}$ include all greatest lower bounds of pairs of their respective members.
  This can be done in the same way as it was done in the proof of Lemma~\ref{lem:nufptfilter} as that proof is constructive insofar concerned with greatest lower bounds.
\end{proof}
This lemma, together with Lemma~\ref{lem:nuxpfilter} and Lemma~\ref{lem:nufptfilter} are possible because of the way we have defined the orders on parameterizations.
In other words, these lemmas justify the definitions of the nonuniform and uniform order on parameterizations.

For the semi-uniform filter, we have an analogue of Theorem~\ref{thm:nufptprincipal}.
\begin{theorem}
\label{thm:sufptprincipal}
  For any set $A$ that is \immune{\cl{P}}, the filter $\calF^{\quasile_\rmnu}_{(A, \clsu{FPT})}$ is principal in $\calL^{\quasile_\rmnu}_\clsu{FPT}$.
\end{theorem}
\begin{proof}
  Identical to the proof of Theorem~\ref{thm:nufptprincipal}.
\end{proof}

For being principal in $\calL^\quasile_\cl{FPT}$, we would need an additional computable upper bound to the gap between parameterizations.
Still, a parameterization with imix would be below one without it even in the nonuniform order.
That is, in accordance with Lemma~\ref{lem:imix}, the nonuniform order on parameterizations distinguishes parameterizations with imix from those without.
Therefore, the corollary to Theorem~\ref{thm:nufptprincipal} works for the uniform case too.
The definitions of levelability and immunity with respect to semi-uniform and uniform parameterized complexity classes are the same as those with respect to nonuniform ones, with the sole difference that for uniform complexity classes they now refer to the uniform order on parameterizations.
\begin{corollary}
\label{cor:fptimmune}
  If a set is \immune{\cl{P}}, it is \immune{\clsu{FPT}} and \immune{\cl{FPT}}.
\end{corollary}

Of course, we \emph{can} show that filters with respect to \cl{FPT} are principal for sets in \cl{P}.
\begin{theorem}
\label{thm:fptprincipal}
  For any set $A$ that is in \cl{P}, the filter $\calF^\quasile_{(A, \cl{FPT})}$ is principal in $\calL^\quasile_\cl{FPT}$.
\end{theorem}
\begin{proof}
  Regardless of the parameter space $\Omega$, the full parameterization $\binary^+ \times \Omega$ is one with which $A$ is in \cl{FPT}.
  Since the class of this parameterization is the least element of the lattice of parameterizations, the filter must be principal.
\end{proof}

Whether or not there exists a set outside \cl{P} for which the filter with respect to \cl{FPT} is principal is an open problem.
We shall see that any such set is necessarily \immune{\cl{P}}.
As before, the converse of Corollary~\ref{cor:fptimmune} does not hold and not all \levelable{\cl{P}} sets are \levelable{\clsu{FPT}} or \levelable{\cl{FPT}}.
\begin{theorem}
  There are \levelable{\cl{P}} sets that are not \levelable{\clsu{FPT}} and not \levelable{\cl{FPT}}.
\end{theorem}
\begin{proof}
  The proof of Theorem~\ref{thm:nonnufptlevelable} is constructive and applies in the semi-uniform and uniform situations too.
\end{proof}

The existence of \levelable{\clsu{FPT}} sets and \levelable{\cl{FPT}} sets can be shown similarly to how the existence of \levelable{\clnu{FPT}} sets was shown.
\begin{theorem}
\label{thm:fptlevelable}
  Every set outside \cl{P} from which there is a \lli{} reduction to itself is \levelable{\clsu{FPT}} and \levelable{\cl{FPT}}.
\end{theorem}
\begin{proof}
  The proof of Theorem~\ref{thm:nufptlevelable} is constructive and applies in the semi-uniform and uniform situations too.
\end{proof}

Thus, by Lemma~\ref{lem:llireduction}, there exist \levelable{\clsu{FPT}} sets and \levelable{\cl{FPT}} sets.
Again, it could well be that some sets are neither \immune{\clsu{FPT}} nor \levelable{\clsu{FPT}} (and similarly for \cl{FPT}).
By Corollary~\ref{cor:nufptimmune}, every \immune{\cl{P}} set is \immune{\clsu{FPT}} and \immune{\cl{FPT}}.
Conversely, because every \cl{P}-slice for a decidable set can be made part of a parameterization with which the set is fixed-parameter tractable, being \levelable{\clsu{FPT}} or \levelable{\cl{FPT}} implies being \levelable{\cl{P}}.
This can be graphically depicted as in Figure~\ref{fig:partition}.
The same holds for \clsu{FPT}.
\begin{figure}[htb]
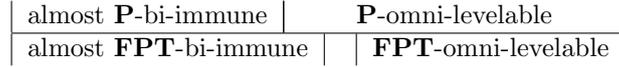

  \centering
  \begin{tabular}{|cccccc|}
    \multicolumn{2}{|c|}{\immune{\cl{P}}}	& \multicolumn{4}{|c|}{\levelable{\cl{P}}} \\
    \hline
    \multicolumn{3}{|c|}{\immune{\cl{FPT}}}	&	& \multicolumn{2}{|c|}{\levelable{\cl{FPT}}}
  \end{tabular}
  \caption{
    The universe of sets, represented by the horizontal line, can be divided according to levelability with respect to \cl{P} or with respect to \cl{FPT}.
    Note that \cl{P} is part of the \immune{\cl{P}} sets.
  }
  \label{fig:partition}
\end{figure}

As we did for \clnu{FPT}, we can answer the principality question for all sets that are not \levelable{\clsu{FPT}} or \levelable{\cl{FPT}}.
\begin{theorem}
\label{thm:sufptnonprincipal}
  For any set $A$ that is \levelable{\cl{P}} and not \levelable{\clsu{FPT}} (\levelable{\cl{FPT}}), the filter $\calF^{\quasile_\rmnu}_{(A, \clsu{FPT})}$ ($\calF^\quasile_{(A, \cl{FPT})}$) is nonprincipal.
\end{theorem}
\begin{proof}
  Identical to the proof of Theorem~\ref{thm:nufptnonprincipal}.
\end{proof}

In contrast to the nonuniform setting we can extend Theorem~\ref{thm:sufptnonprincipal} proving that the filters produced by \levelable{\clsu{FPT}} sets and \levelable{\cl{FPT}} sets are also nonprincipal.
\begin{theorem}
\label{thm:fptnonprincipal}
  For any set $A$ that is not \immune{\clsu{FPT}} (\immune{\cl{FPT}}), the filter $\calF^{\quasile_\rmnu}_{(A, \clsu{FPT})}$ ($\calF^\quasile_{(A, \cl{FPT})}$) is nonprincipal.
\end{theorem}
\begin{proof}
  We present a proof by contradiction for the uniform setting.
  The semi-uniform setting is subsumed in this proof.
  Assuming $\calF^\quasile_{(A, \cl{FPT})}$ is the principal filter induced by a parameterization $\eta \subseteq \binary^+ \times \Omega$ we construct a parameterization $\eta'$ such that we have $(A, \eta') \in \cl{FPT}$ and $\eta \nquasile \eta'$.

  Let $c$ and $\Psi$ be the constant and the Turing machine witnessing that $A$ is in \cl{FPT} with parameterization $\eta$ and let $\omega_1, \omega_2, \ldots$ be an effective enumeration of (an unbounded subset of) $\Omega$.
  Denote by $\psi_{\omega_i}$ the partial application of $\Psi$ to $\omega_i$.
  Thus, $\psi_{\omega_i}$ is a \clO{n^c}-approximation for $A$ with domain $\eta_{\omega_i}$.
  Consider the following partial decision procedure for $A$, uniformly defined for all $k \in \bbN$.
  On input $x$, the procedure takes the following steps.
  \begin{enumerate}
  \item
    Spend $k \cdot \length{x}^{c + 2}$ time computing the values $\psi_{\omega_1}(x), \psi_{\omega_2}(x), \ldots$.
  \item
    If any computed value is either $1$ or $0$, return it, otherwise return nothing.
  \end{enumerate}
  This partial decision procedure is a \clO{n^{c + 2}}-approximation for $A$ and moreover it defines a parameterization $\eta' \subseteq \binary^+ \times \bbN$ with which $A$ is in \cl{FPT} when we interpret $k$ as the parameter value.

  Because the exponent $c + 2$ is more than necessary to simulate any fixed number of \clO{n^c}-approximations for $A$, we are able to compute an increasing number of values $\psi_{\omega_1}, \psi_{\omega_2}, \ldots$ with increasing $\length{x}$ for any constant value of $k$.
  Because $\eta$ necessarily has imix, this enables us to decide membership for elements from arbitrary high slices in $\eta$.
  Hence the gap between $\eta$ and $\eta'$ is always infinite, proving $\eta \nquasile_\rmnu \eta'$ and also $\eta \nquasile \eta'$.
\end{proof}

Combining the above theorem with Theorem~\ref{thm:sufptprincipal} and Theorem~\ref{thm:sufptnonprincipal}, we get a complete picture of which principal filters in $\calL^{\quasile_\rmnu}_\clsu{FPT}$ can occur as filters of the form $\calF^{\quasile_\rmnu}_{(A, \clsu{FPT})}$ with some set $A$.
\begin{corollary}
  For any set $A$, the filter $\calF^{\quasile_\rmnu}_{(A, \clsu{FPT})}$ is principal if and only if $A$ is \immune{\cl{P}}.
\end{corollary}
For the uniform filters, we lack a uniform equivalent to Theorem~\ref{thm:sufptprincipal}.
However, Theorem~\ref{thm:fptnonprincipal} holds that whenever $\calF^\quasile_{(A, \cl{FPT})}$ is principal, a principal parameterization of that filter does not have imix.
Conversely, every parameterization that has imix does not occur as a principal parameterization in any filter of the form $\calF^\quasile_{(A, \cl{FPT})}$.
As parameterizations with imix are the most interesting from an applications point of view, we find that sets for which parameterized algorithms are attractive do not admit optimal parameterizations.

The proof of Theorem~\ref{thm:fptnonprincipal} has a clear kinship to that of the time hierarchy theorem~\citep{hartmanis1965computational}.
Observe though, that the time hierarchy theorem constitutes a hierarchy of \emph{sets}, whereas the current theorem is about a hierarchy of \emph{algorithms}.
The proof of Theorem~\ref{thm:fptnonprincipal} is made possible in its current form by our choice of Turing machines as a model of computation.
For other machine models, random access machines for example, different, but at the same time similar, proofs can be given.
To make the paper entirely machine model independent, we would have to take the jumps in the time bounds a bit bigger, though still polynomial of small degree as warranted by the sequential computation thesis~\Citep{emdeboas2014machine}, also known as the extended Church--Turing thesis~\citep{parberry1986parallel}.

We have made use of the time hierarchy theorem to break out of the constant exponent in the running time of fixed-parameter tractable algorithms.
When the exponent is allowed to vary, the diagonalization in the proof of Theorem~\ref{thm:fptnonprincipal} fails.
An immediate class that allows such variations in the exponent of the polynomial running time is the semi-uniform variant of \clXnu{P}.
\begin{definition}
  A set $A$ is in \clXsu{P} with parameterization $\eta$ if there is a Turing machine $\Psi$ taking two inputs, such that for every parameter value $k$ the partial application of $\Psi$ to $k$ yields an \cl{P}-approximation for $A$ with domain $\eta_k$.
\end{definition}

Not only does the proof of Theorem~\ref{thm:fptnonprincipal} fail on \clXsu{P}, the proof of Theorem~\ref{thm:nuxpprincipal} fails on \clXsu{P} as well.
The enumeration of slices in that proof is in general not effective and the resulting parameterization need not put the set under consideration in \clXsu{P}.
When focusing only on parameterizations that provably put a set in \clXsu{P}, we regain something akin to Theorem~\ref{thm:nuxpprincipal}.
\begin{theorem}
\label{thm:xpprincipal}
  Given any formal proof system, for any set $A$ there is a least parameterization among those provably, via a witnessing Turing machine, putting $A$ in \clXsu{P}.
\end{theorem}
\begin{proof}
  Let $\mathfrak{F}$ be the fixed deductive system.
  A form of universal search, along the lines of \citet{hutter2002fastest}, through \cl{P}-approximations is possible in the parameterized setting.
  Consider the following partial decision procedure for $A$, uniformly defined for all $k \in \bbN$.
  On input $x$, the procedure takes the following steps.
  \begin{enumerate}
  \item
    \begin{enumerate}
    \item
      Set $M \deq \emptyset$.
    \item
      For all proofs $P$ in $\mathfrak{F}$ and all machines $\Phi$ both of length at most $k$ do:
      \\\-\quad If $P$ proves that $\Phi$ is a \cl{P}-approximation for $A$ then:
      \\\-\quad\quad Set $M = M \cup \{\Phi\}$.
    \end{enumerate}
  \item
    \begin{enumerate}
    \item
    \label{alg:xp:sim:sim}
      For $\Phi$ in $M$ do:
      \\\-\quad If $\Phi(x) \in \{0, 1\}$ then return $\Phi(x)$.
    \item
      return nothing.
    \end{enumerate}
  \end{enumerate}
  Since, for all values of $k$, the set $M$ is finite throughout the execution of this partial decision procedure, step \ref{alg:xp:sim:sim} can be executed in polynomial time and the procedure is a \cl{P}-approximation for $A$.
  Moreover, when we interpret $k$ as the parameter value, the above procedure defines a parameterization $\eta$ with which $A$ is in \clXsu{P}.

  It remains to show that $\eta$ lies below every other parameterization that provably puts $A$ in \clXsu{P}.
  Let $\Psi$ be the witness Turing machine corresponding to some parameterization $\eta' \subseteq \binary^+ \times \Omega$ with which $A$ is in \clXsu{P}.
  For any parameter value $\omega \in \Omega$, the slice $\eta'_\omega$ is included in the slice $\eta_{\length{\langle \Psi, \omega\rangle} + \bigO(1)}$, where the hidden constant depends on the length of an $\mathfrak{F}$-proof of $A$ being in \clXsu{P} with parameterization $\eta'$ and the overhead of a program applying $\Psi$ to $\omega$.
  Hence we have $\eta \quasile \eta'$ as desired.
\end{proof}

Thus, adding a provability requirement offsets the limitations we incurred by moving to a uniform setting.
The provability requirement in Theorem~\ref{thm:xpprincipal} enforces the effectiveness that was not present in the proof of Theorem~\ref{thm:nuxpprincipal}.
Additionally, from the existence of least parameterizations that provably put a set in \clXsu{P} it follows that sets are either \immune{\clXsu{P}} or \levelable{\clXsu{P}}.
\begin{corollary}
\label{cor:xpneither}
  When restricting to parameterizations that provably put a set in \clXsu{P}, sets are either \immune{\clXsu{P}} or \levelable{\clXsu{P}}.
\end{corollary}

While a set is \levelable{\cl{P}} precisely when it is \levelable{\clXnu{P}}, it could be that the \levelable{\clXsu{P}} sets form a strict subset of the \levelable{\cl{P}} sets.

Without accounting for proofs, the partial decision procedure constructed in the proof of Theorem~\ref{thm:xpprincipal} would consider only the lengths of the \cl{P}-approximations $\Phi$.
Then, parameter values are much akin to instance complexity \citep{orponen1994instance}.
\begin{definition}
  The time $t$ bounded \emph{instance complexity} of a string $x$ with respect to a set $A$ is
  \begin{equation*}
    \ic^t(x : A) \deq \min\{\length{\Phi} \st \text{$\Phi$ is a $t$-approximation for $A$ and $x \in \dom(\Phi)$}\},
  \end{equation*}
  where a $t$-approximation for $A$ is a partial decision procedure for $A$ that runs in time $t$.
\end{definition}

The following theorem gives a time bounded version of Lemma~6.1 in~\citep{witteveen2016fpd}.

\begin{theorem}
\label{thm:nufptic}
  For any polynomial $p$, any set $A$ is in \clnu{FPT} with parameterization $\eta \deq \{(x, k) \st \ic^p(x : A) \le k\}$.
\end{theorem}
\begin{proof}
  Let $c$ be such that $p$ is in $\bigO(n^c)$.
  For any value of $k$, there are only finitely many $p$-approximations for $A$ of length at most $k$.
  These can be combined into a single \clO{n^c}-approximation for $A$ for which the domain is exactly $\eta_k$.
  Hence, $A$ is in \clnu{FPT} with parameterization $\eta$.
\end{proof}

Observe how this theorem is similar to Theorem~\ref{thm:nuxpprincipal}.
Stretching our definition of time bounded instance complexity a little, we can define a parameterization $\{(x, k) \st \ic^\cl{P}(x : A) \le k\}$.
The previous theorem can be adapted to show that $A$ is in \clXnu{P} with this parameterization and as seen in the proof of Theorem~\ref{thm:nuxpprincipal} this parameterization is a principal parameterization in $\calF^{\quasile_\rmnu}_{(A, \clXnu{P})}$.

Based on Theorem~\ref{thm:xpprincipal} and Theorem~\ref{thm:nufptic}, for a parameterization $\eta$ that puts a set in \cl{FPT}, we think of $\mu_\eta$ as conveying a sense of complexity.
Unlike instance complexity, the sense of complexity conveyed by $\mu_\eta$ can be called \emph{uniform} as it requires a uniform way of deriving approximations.
Indeed, there is a sort of converse to Theorem~\ref{thm:nufptic}.
\begin{theorem}
  For any set $A$ that is in \cl{FPT} with parameterization $\eta \subseteq \binary^+ \times \Omega$ there is a function $f$ and polynomial $p$ such that for all $x$ we have, up to an additive constant independent of $x$,
  \begin{equation*}
    \ic^{f(\mu_\eta(x))p}(x : A) \le \mu_\eta(x).
  \end{equation*}
\end{theorem}
\begin{proof}
  It suffices to show that for some $c$ and every $\omega \in \Omega$ there is a \clO{n^c}-approximation for $A$ of which the size is, up to an additive constant, bounded by $\length{\omega}$.
  Let $\Psi$ be the Turing machine witnessing that $A$ is in \cl{FPT} with parameterization $\eta$.
  By definition, there is a $c$ such that for every value $\omega$ the partial application of $\Psi$ to $\omega$ yields a \clO{n^c}-approximation for $A$.
  As this \clO{n^c}-approximation can be constructed from $\Psi$ and $\omega$, the length of its specification can, up to an independent additive constant, bounded by $\length{\langle \Psi, \omega\rangle}$.
  Because $\Psi$ is fixed for all $x$, the theorem follows.
\end{proof}

The behavior of the complexity measure embodied, for a parameterization $\eta$, by $\mu_\eta$ can be somewhat untangible.
When $\eta$ is not a principal parameterization for a given set, there are sharper parameterizations and hence sharper complexity measures possible with respect to that set.
Many sets, however, do not allow for principal parameterizations at all.
Furthermore, even when a parameterization $\eta$ is principal, there can be other parameterizations equal to it that give rise to a measure of complexity lower than $\mu_\eta$.
This improvement though, is of a bounded nature and a sense of optimality is still given to the complexity behavior of principal parameterizations.

Somewhat more abstract, for a set $A$ the filter $\calF^\quasile_{(A, \cl{FPT})}$ itself can be considered a representation of the distribution of complexity of instances with respect to $A$.
This view has the added benefit that it is applicable not only when the filter is principal.
As we did previously in the nonuniform context, we shall classify sets based on their filters with respect to  \cl{FPT}.
Where polynomial isomorphism of sets indicates a \emph{comparable} distribution of difficulty, having the same filter with respect to \cl{FPT} signifies that \emph{exactly the same} elements are difficult.
\begin{theorem}
  For any set $X$ in \cl{P} and any set $A$ we have
  \begin{equation*}
    \calF^\quasile_{(A, \cl{FPT})} = \calF^\quasile_{(A \symdiff X, \cl{FPT})}.
  \end{equation*}
\end{theorem}
\begin{proof}
  This follows from the fact that given $X$, the approximations for $A$ corresponding to a parameterization can be uniformly transformed into approximations for $A \symdiff X$ with the same domain, and vice versa.
\end{proof}

Compared to Theorem~\ref{thm:nufptsymdiffeq}, we needed the additional remark that the transformation used in the proof is uniform in the parameter.
With a similar addition we retrieve a version of Theorem~\ref{thm:nufptsymdiffsubeq} for \cl{FPT}.
\begin{theorem}
  For any two sets $A, B$ satisfying $\calF^\quasile_{(A, \cl{FPT})} = \calF^\quasile_{(B, \cl{FPT})}$ we have
  \begin{equation*}
    \calF^\quasile_{(A, \cl{FPT})} \subseteq \calF^\quasile_{(A \symdiff B, \cl{FPT})}.
  \end{equation*}
\end{theorem}
\begin{proof}
  Given a parameterization $\eta$ with which $A$ and $B$ are in \cl{FPT}, let $\Psi_A$ and $\Psi_B$ be Turing machines that yield the respective approximations on partial application to a parameter value.
  From $\Psi_A$ and $\Psi_B$ it is possible to define an approximation for $A \symdiff B$ uniformly in a parameter value $k$.
  On input $x$, when $(x, k)$ is a member of $\eta$ this approximation simply outputs the exclusive disjunction of $\Psi_A(k, x)$ and $\Psi_B(k, x)$.
  It follows that $A \symdiff B$ is also in \cl{FPT} with parameterization $\eta$.
\end{proof}

We conjecture that the filter of the symmetric difference collapses to that of a set in \cl{P}.
\begin{conjecture}
  For any two sets $A, B$ we have
  \begin{equation*}
    \calF^\quasile_{(A, \cl{FPT})} = \calF^\quasile_{(B, \cl{FPT})} \iff \calF^\quasile_{(A \symdiff B, \cl{FPT})} = \calF^\quasile_{(\emptyset, \cl{FPT})}.
  \end{equation*}
\end{conjecture}
As with Conjecture~\ref{con:nusymdiff}, this conjecture implies a separation result.
However, uniformity constraints make this separation result more intricate than that obtained from Conjecture~\ref{con:nusymdiff}.

\section{Conclusion}

We have explored the algebraic structure of parameterizations underlying the parameterized analysis of computational complexity.
Under the general definitions of parameter spaces and parameterizations of Section~\ref{sec:parameterspaces}, parameterizations generate a lattice.
We found that for several parameterized complexity classes the parameterizations that put a given set in that parameterized complexity class form a filter in this lattice.
These filters hold information about the complexity make up of the sets that produce them.
Thus, in the examination of sets these filters can act as a proxy to their internal distribution of complexity.

Solely based on the filters induced by given sets with respect to parameterized complexity classes we could extend the classifications ``\immune{\cl{P}}'' and ``\levelable{\cl{P}}'' into the parameterized context.
We have seen that the classical classifications coincide with the parameterized classifications ``\immune{\clXnu{P}}'' and ``\levelable{\clXnu{P}}''.
Furthermore we have seen that for a parameterized complexity class $\cl{C} \in \{\clXnu{P},\allowbreak \clXsu{P},\allowbreak \clnu{FPT},\allowbreak \clsu{FPT},\allowbreak \cl{FPT}\}$ there are sets that are \immune{\cl{C}} as well as sets that are \levelable{\cl{C}}.
In particular, when we denote the class of sets outside \cl{P} that admit a \lli{} reduction to itself by \cl{LLI}, we have the following relations between classifications of sets.
\begin{align*}
  \text{in \cl{P}}	&\implies \text{\immune{\cl{P}}} \implies \text{\immune{\cl{C}}}
\shortintertext{and}
  \text{in \cl{LLI}}	&\implies \text{\rlap{\levelable{\cl{C}}}\phantom{\immune{\cl{P}}}} \implies \text{\levelable{\cl{P}}},
\end{align*}
where the last implications on each of these two lines become equivalences when \cl{C} is \clXnu{P}.
Indeed, there exists \immune{\cl{C}} sets and \levelable{\cl{C}} sets because \cl{P} is nonempty and, by Lemma~\ref{lem:llireduction}, \cl{LLI} is nonempty.
These implications expand on Figure~\ref{fig:partition} and it should be noted that no set can be both \immune{\cl{C}} and \levelable{\cl{C}}, and with respect to $\cl{C} \in \{\clnu{FPT}, \clsu{FPT}, \cl{FPT}\}$ it is not ruled out that some set is neither.

In defining levelability with respect to classes of fixed-parameter tractable sets we have made explicit the role of the exponent in the involved polynomial running time bounds.
Although the alternative characterizations of \clX{P} and \cl{FPT} on page~\pageref{eq:xpfpt} provide a good motivation for doing so, we wonder whether any set would meet the naive definition of being \levelable{\cl{FPT}}.
\begin{openproblem}
  Is there a set $A$ and a parameterization $\eta$ with which $A$ is in \cl{FPT} such that every parameterization $\eta' \quasile \eta$ with which $A$ is in \cl{FPT} has imix?
\end{openproblem}

Another aspect of parameterized complexity that is captured by filters for sets with respect to parameterized complexity classes is the existence of optimal parameterizations.
A parameterization that puts a set in one of our parameterized complexity classes is optimal when it is below all other such parameterizations in the order relevant for the particular parameterized complexity class.
It should be noted that the improvement signified by one parameterization being below another is far stronger than the improvements related to typical races in parameterized algorithmics, where a parameterization is held fixed.
For our nonuniform and semi-uniform parameterized complexity classes we use the nonuniform order on parameterizations.
For our strongly uniform parameterized complexity class the uniform order on parameterizations is the most natural.
An optimal parameterization for a set with respect to a parameterized complexity class then exists when the induced filter with respect to that class is principal.
Our results on principality of filters with respect to the different parameterized complexity classes can be summarized as in Table~\ref{tab:principality}.
\begin{table}[htb]
  \centering
  \begin{tabular}{l|cc}
    	& \emph{\immune{\cl{P}}}	& \emph{\levelable{\cl{P}}} \\
    \hline
    \clXnu{P}	& principal	& principal \\
    \clXsu{P} (provably)	& principal	& principal \\
    \clnu{FPT}	& principal	& Theorem~\ref{thm:nufptnonprincipal} \\
    \clsu{FPT}	& principal	& nonprincipal \\
    \cl{FPT}	& Theorem~\ref{thm:fptprincipal}	& nonprincipal \\
  \end{tabular}
  \caption{Principality of filters with respect to parameterized complexity classes depending on the classification of the set inducing the filter.}
  \label{tab:principality}
\end{table}

Every value in a cell of Table~\ref{tab:principality} is backed by one or more theorems or corollaries.
By extending the uniformity constraints of \clXsu{P} with provability constraints, we were able to obtain the same principality results as for \clXnu{P}.
Regarding fixed-parameter tractability we were able to obtain a necessary and sufficient condition for the principality of filters in the semi-uniform case.
A filter for a set with respect to \clsu{FPT} is principal if and only if the set is \immune{\cl{P}}.
Consequently, no \levelable{\clsu{FPT}} set has an optimal parameterization with respect to \clsu{FPT} as the \levelable{\clsu{FPT}} sets are all \levelable{\cl{P}}.
This is noteworthy because from a practical point of view a parameterized approach is only worthwhile for sets that are \levelable{\clsu{FPT}}.
In many cases attention is limited even further to only strongly uniform parameterized algorithms.
Such algorithms are of practical use in particular for \levelable{\cl{FPT}} sets.
As \levelable{\cl{FPT}} are again also \levelable{\cl{P}}, in Table~\ref{tab:principality} we see that no \levelable{\cl{FPT}} set admits an optimal parameterization with respect to \cl{FPT}.

When the filter induced by a set is nonprincipal, it is not possible to capture the structure responsible for the computational hardness of the set by a parameterization.
In that case, there are infinitely many distinct structural properties an element may have that can be used to decide membership of the element in the set in a way that defies the computational hardness of the set in general.

For two cells in Table~\ref{tab:principality} our results are incomplete.
The filter of a \levelable{\cl{P}} set with respect to \clnu{FPT} is nonprincipal when the set is not \levelable{\clnu{FPT}}.
However, for \levelable{\clnu{FPT}} sets the principality of the corresponding filters with respect to \clnu{FPT} is still unknown.
\begin{openproblem}
  Is there a \levelable{\clnu{FPT}} set $A$ for which $\calF^{\quasile_\rmnu}_{(A, \clnu{FPT})}$ is principal?
\end{openproblem}

We have a similar incomplete picture for the filter of an \immune{\cl{P}} set with respect to \cl{FPT}.
Such a filter is known to be principal for the rather trivial case when the set is in \cl{P}.
Otherwise we only know that filters of \levelable{\cl{P}} sets with respect to \cl{FPT} are nonprincipal.
\begin{openproblem}
  Is there a set $A$ outside \cl{P} for which $\calF^\quasile_{(A, \cl{FPT})}$ is principal?
\end{openproblem}

\bibliography{\jobname}

\begin{thebibliography}{30}
\providecommand{\natexlab}[1]{#1}
\providecommand{\url}[1]{\texttt{#1}}
\expandafter\ifx\csname urlstyle\endcsname\relax
  \providecommand{\doi}[1]{doi: #1}\else
  \providecommand{\doi}{doi: \begingroup \urlstyle{rm}\Url}\fi

\bibitem[Arora and Barak(2009)]{arora2009computational}
Sanjeev Arora and Boaz Barak.
\newblock \emph{Computational complexity: a modern approach}.
\newblock Cambridge University Press, 2009.

\bibitem[Balc{\'a}zar and Sch{\"o}ning(1985)]{balcazar1985bi}
Jos{\'e}~L. Balc{\'a}zar and Uwe Sch{\"o}ning.
\newblock Bi-immune sets for complexity classes.
\newblock \emph{Theory of Computing Systems}, 18\penalty0 (1):\penalty0 1--10,
  1985.

\bibitem[Bodlaender et~al.(2013)Bodlaender, Cygan, Kratsch, and
  Nederlof]{bodlaender2013deterministic}
Hans~L. Bodlaender, Marek Cygan, Stefan Kratsch, and Jesper Nederlof.
\newblock Deterministic single exponential time algorithms for connectivity
  problems parameterized by treewidth.
\newblock In \emph{International Colloquium on Automata, Languages, and
  Programming}, pages 196--207, 2013.

\bibitem[Book et~al.(1988)Book, Du, and Russo]{book1988polynomial}
Ronald~V. Book, D-Z Du, and David~A. Russo.
\newblock On polynomial and generalized complexity cores.
\newblock In \emph{Proceedings of the third annual Structure in Complexity
  Theory conference}, pages 236--250. IEEE, 1988.

\bibitem[Cook(1971)]{cook1971complexity}
Stephen~A. Cook.
\newblock The complexity of theorem-proving procedures.
\newblock In \emph{Proceedings of the third annual ACM Symposium on Theory of
  Computing}, pages 151--158. ACM, 1971.

\bibitem[Davey and Priestley(2002)]{davey2002introduction}
Brian~A. Davey and Hilary~A. Priestley.
\newblock \emph{Introduction to lattices and order}.
\newblock Cambridge University Press, 2002.

\bibitem[Downey and Fellows(1999)]{downey1999parameterized}
Rodney~G. Downey and Michael~R. Fellows.
\newblock \emph{Parameterized Complexity}.
\newblock Springer, 1999.

\bibitem[{\sortkey{Emde Boas}}van Emde~Boas(2014)]{emdeboas2014machine}
Peter {\sortkey{Emde Boas}}van Emde~Boas.
\newblock Machine models and simulations.
\newblock \emph{Handbook of Theoretical Computer Science, volume A}, pages
  1--66, 2014.

\bibitem[Fellows et~al.(2013)Fellows, Jansen, and Rosamond]{fellows2013towards}
Michael~R. Fellows, Bart~M.P. Jansen, and Frances Rosamond.
\newblock Towards fully multivariate algorithmics: Parameter ecology and the
  deconstruction of computational complexity.
\newblock \emph{European Journal of Combinatorics}, 34\penalty0 (3):\penalty0
  541--566, 2013.

\bibitem[Flum and Grohe(2003)]{flum2003describing}
J{\"o}rg Flum and Martin Grohe.
\newblock Describing parameterized complexity classes.
\newblock \emph{Information and Computation}, 187\penalty0 (2):\penalty0
  291--319, 2003.

\bibitem[Flum and Grohe(2006)]{flum2006parameterized}
J{\"o}rg Flum and Martin Grohe.
\newblock \emph{Parameterized Complexity Theory}.
\newblock Springer, 2006.

\bibitem[Garey and Johnson(1979)]{garey1979computers}
Michael~R. Garey and David~S. Johnson.
\newblock \emph{Computers and Intractability: A Guide to the Theory of
  {NP}-Completeness}.
\newblock Freeman, 1979.

\bibitem[Garg and Philip(2016)]{garg2016raising}
Shivam Garg and Geevarghese Philip.
\newblock Raising the bar for vertex cover: fixed-parameter tractability above
  a higher guarantee.
\newblock In \emph{Proceedings of the twenty-seventh annual ACM-SIAM Symposium
  on Discrete Algorithms}, pages 1152--1166. Society for Industrial and Applied
  Mathematics, 2016.

\bibitem[Goldreich(2008)]{goldreich2008computational}
Oded Goldreich.
\newblock Computational complexity: a conceptual perspective.
\newblock \emph{ACM SIGACT News}, 39\penalty0 (3):\penalty0 35--39, 2008.

\bibitem[Hartmanis and Stearns(1965)]{hartmanis1965computational}
Juris Hartmanis and Richard~E. Stearns.
\newblock On the computational complexity of algorithms.
\newblock \emph{Transactions of the American Mathematical Society},
  117:\penalty0 285--306, 1965.

\bibitem[Hutter(2002)]{hutter2002fastest}
Marcus Hutter.
\newblock The fastest and shortest algorithm for all well-defined problems.
\newblock \emph{International Journal of Foundations of Computer Science},
  13\penalty0 (3):\penalty0 431--443, 2002.

\bibitem[Ko and Moore(1981)]{ko1981completeness}
Ker-I Ko and Daniel Moore.
\newblock Completeness, approximation and density.
\newblock \emph{SIAM Journal on Computing}, 10\penalty0 (4):\penalty0 787--796,
  1981.

\bibitem[Komusiewicz and Niedermeier(2012)]{komusiewicz2012new}
Christian Komusiewicz and Rolf Niedermeier.
\newblock New races in parameterized algorithmics.
\newblock In \emph{Mathematical Foundations of Computer Science}, volume~12,
  pages 19--30. Springer, 2012.

\bibitem[Li and Vit{\'a}nyi(1997)]{li1997introduction}
Ming Li and Paul Vit{\'a}nyi.
\newblock \emph{An introduction to Kolmogorov complexity and its applications}.
\newblock Springer, 1997.

\bibitem[Lynch(1975)]{lynch1975reducibility}
Nancy Lynch.
\newblock On reducibility to complex or sparse sets.
\newblock \emph{Journal of the ACM}, 22\penalty0 (3):\penalty0 341--345, 1975.

\bibitem[Odifreddi(1992)]{odifreddi1992classical}
Piergiorgio Odifreddi.
\newblock \emph{Classical Recursion Theory}.
\newblock Studies in Logic and the Foundations of Mathematics. Elsevier, 1992.

\bibitem[Orponen(1986)]{orponen1986classification}
Pekka Orponen.
\newblock A classification of complexity core lattices.
\newblock \emph{Theoretical Computer Science}, 47:\penalty0 121--130, 1986.

\bibitem[Orponen et~al.(1985)Orponen, Russo, and
  Sch{\"o}ning]{orponen1985polynomial}
Pekka Orponen, David~A. Russo, and Uwe Sch{\"o}ning.
\newblock Polynomial levelability and maximal complexity cores.
\newblock In \emph{International Colloquium on Automata, Languages, and
  Programming}, pages 435--444. Springer, 1985.

\bibitem[Orponen et~al.(1986)Orponen, Russo, and
  Sch{\"o}ning]{orponen1986optimal}
Pekka Orponen, David~A. Russo, and Uwe Sch{\"o}ning.
\newblock Optimal approximations and polynomially levelable sets.
\newblock \emph{SIAM Journal on Computing}, 15\penalty0 (2):\penalty0 399--408,
  1986.

\bibitem[Orponen et~al.(1994)Orponen, Ko, Sch{\"o}ning, and
  Watanabe]{orponen1994instance}
Pekka Orponen, Ker-I Ko, Uwe Sch{\"o}ning, and Osamu Watanabe.
\newblock Instance complexity.
\newblock \emph{Journal of the ACM}, 41\penalty0 (1):\penalty0 96--121, 1994.

\bibitem[Papadimitriou(2003)]{papadimitriou2003computational}
Christos~H. Papadimitriou.
\newblock Computational complexity.
\newblock In \emph{Encyclopedia of Computer Science}, pages 260--265. Wiley,
  2003.

\bibitem[Parberry(1986)]{parberry1986parallel}
Ian Parberry.
\newblock Parallel speedup of sequential machines: A defense of parallel
  computation thesis.
\newblock \emph{ACM SIGACT News}, 18\penalty0 (1):\penalty0 54--67, 1986.

\bibitem[Sch{\"o}ning(1989)]{schoning1989probabilistic}
Uwe Sch{\"o}ning.
\newblock Probabilistic complexity classes and lowness.
\newblock \emph{Journal of Computer and System Sciences}, 39\penalty0
  (1):\penalty0 84--100, 1989.

\bibitem[Witteveen and Torenvliet(2016)]{witteveen2016fpd}
Jouke Witteveen and Leen Torenvliet.
\newblock Fixed-parameter decidability: Extending parameterized complexity
  analysis.
\newblock \emph{Mathematical Logic Quarterly}, 62\penalty0 (6):\penalty0
  596--607, 2016.

\bibitem[Young(1983)]{young1983some}
Paul Young.
\newblock Some structural properties of polynomial reducibilities and sets in
  {NP}.
\newblock In \emph{Proceedings of the fifteenth annual ACM Symposium on Theory
  of Computing}, pages 392--401. ACM, 1983.

\end{thebibliography}
\end{document}